\documentclass[11pt,letterpaper,fleqn]{article}

\usepackage{amsmath,amssymb,amsthm}
\usepackage{thm-restate}
\usepackage[margin=1in]{geometry}
\usepackage{xcolor}
\usepackage{graphicx}
\usepackage{bm}
\ifx\compatmode\undefined
\usepackage{bbm}
\else
\fi
\usepackage{comment}
\usepackage{subcaption}
\usepackage{tikz}
\usetikzlibrary{arrows.meta,patterns}
\usepackage{enumitem}
\usepackage{authblk}
\usepackage{xspace}
\usepackage[hypertexnames=false,bookmarksnumbered=true,final]{hyperref}
\usepackage[capitalize,sort]{cleveref}

\def\colorschemesepia{sepia}
\def\colorschemedark{dark}
\def\colorschemelight{light}

\ifx\colorscheme\undefined
\let\colorscheme\colorschemelight
\fi

\ifx\colorscheme\colorschemelight
\colorlet{textColor}{black}
\colorlet{bgColor}{white}
\fi

\ifx\colorscheme\colorschemesepia
\definecolor{textColor}{HTML}{433423}
\definecolor{bgColor}{HTML}{fbf0da}
\fi

\ifx\colorscheme\colorschemedark
\definecolor{textColor}{HTML}{bdc1c6}
\definecolor{bgColor}{HTML}{202124}
\definecolor{textBlue}{HTML}{8ab4f8}
\definecolor{textRed}{HTML}{f9968b}
\definecolor{textGreen}{HTML}{81e681}
\definecolor{textPurple}{HTML}{c58af9}
\else
\colorlet{textBlue}{blue!50!black}
\colorlet{textRed}{red!50!black}
\colorlet{textGreen}{green!50!black}
\definecolor{textPurple}{HTML}{681da8}
\fi

\ifx\colorscheme\colorschemelight\else
\pagecolor{bgColor}
\color{textColor}
\fi

\hypersetup{colorlinks,linkcolor=textRed,citecolor=textRed,urlcolor=textBlue}

\theoremstyle{plain}
\newtheorem{theorem}{Theorem}
\newtheorem{lemma}{Lemma}

\newtheorem{observation}{Observation}
\theoremstyle{definition}
\newtheorem{definition}{Definition}
\newtheorem{problem}{Problem}

\theoremstyle{remark}
\newtheorem{remark}{Remark}
\newtheorem{example}{Example}

\title{Equilibrium Pricing in Oligopolistic Data Markets}

\author[1]{Bhaskar Ray Chaudhury}
\author[1]{Jugal Garg}
\author[1]{Eklavya Sharma}
\author[1]{Jiaxin Song}
\affil[1]{University of Illinois, Urbana-Champaign}
\affil[ ]{\{braycha, jugal, eklavya2, jiaxins8\}@illinois.edu}
\date{\empty}

\let\eps\varepsilon
\newcommand*{\defeq}{:=}
\newcommand*{\Th}{^{\textrm{th}}}

\newcommand*{\wLoG}{without loss of generality}
\makeatletter
\if@twocolumn
  \newcommand*{\ifTwoCol}[2]{#1}
\else
  \newcommand*{\ifTwoCol}[2]{#2}
\fi
\makeatother

\newcommand*{\boolOne}{\mathbbold{1}}  %
\newcommand*{\vecOne}{\mathbf{1}}
\newcommand*{\vecZero}{\mathbf{0}}
\newcommand*{\vecE}{\mathbf{e}}

\newcommand*{\abs}[1]{\lvert #1 \rvert}
\newcommand*{\norm}[1]{\lVert #1 \rVert}

\DeclareMathOperator*{\E}{\mathbb{E}}
\DeclareMathOperator*{\Var}{Var}
\DeclareMathOperator*{\argmin}{argmin}
\DeclareMathOperator*{\argmax}{argmax}

\DeclareMathOperator*{\maximize}{maximize}
\DeclareMathOperator{\supp}{supp}

\newcommand*{\Ical}{\mathcal{I}}
\newcommand{\G}{\mathcal{G}}

\newcommand*{\chat}{\widehat{c}}

\newcommand*{\phat}{\widehat{p}}
\newcommand*{\qhat}{\widehat{q}}
\newcommand*{\rhat}{\widehat{r}}
\newcommand*{\shat}{\widehat{s}}

\newcommand*{\zhat}{\widehat{z}}
\newcommand*{\ellhat}{\widehat{\ell}}
\newcommand*{\rhohat}{\widehat{\rho}}
\newcommand{\x}{\bm{x}}
\newcommand{\z}{\bm{z}}
\newcommand{\bell}{\bm{\ell}}

\newcommand*{\ktild}{\widetilde{k}}

\newcommand*{\xvec}{\bm{x}}

\newcommand*{\zvec}{\bm{z}}
\newcommand*{\ellvec}{\bm{\ell}}
\newcommand*{\ellvechat}{\widehat{\ellvec}}

\newcommand{\Rge}{\mathbb{R}_{\ge 0}}
\newcommand{\NP}{{\sf NP}\xspace}
\newcommand{\xc}{{\sc X3C}\xspace}
\newcommand{\kf}{{\sc Knapsack-1Fraction}\xspace}

\newif\ifHasSuppl

\allowdisplaybreaks
\DeclareMathAlphabet{\mathbbold}{U}{bbold}{m}{n}

\makeatletter
\g@addto@macro{\UrlBreaks}{%
\do\/%
\do\a\do\b\do\c\do\d\do\e\do\f\do\g\do\h\do\i\do\j\do\k\do\l\do\m%
\do\n\do\o\do\p\do\q\do\r\do\s\do\t\do\u\do\v\do\w\do\x\do\y\do\z%
\do\A\do\B\do\C\do\D\do\E\do\F\do\G\do\H\do\I\do\J\do\K\do\L\do\M%
\do\N\do\O\do\P\do\Q\do\R\do\S\do\T\do\U\do\V\do\W\do\X\do\Y\do\Z%
\do\0\do\1\do\2\do\3\do\4\do\5\do\6\do\7\do\8\do\9%
}
\makeatother

\makeatletter
\@ifpackageloaded{enumitem}{%
\newenvironment*{tightemize}{\begin{itemize}[noitemsep]}{\end{itemize}}%
\newenvironment*{tightenum}{\begin{enumerate}[noitemsep]}{\end{enumerate}}%
}{%
\newenvironment*{tightenum}{\begin{enumerate}}{\end{enumerate}}%
}
\makeatother

\let\citep\cite
\let\citet\cite

\begin{document}

\maketitle

\begin{abstract}
We study equilibrium pricing in oligopolistic data markets with budget-constrained buyers
(e.g., machine learning companies purchasing data to improve model accuracy) and strategic data sellers.
Sellers compete by setting prices for their datasets, giving rise to a pricing game whose pure Nash equilibria correspond to equilibrium prices.
While equilibrium prices are guaranteed for rivalrous goods via competitive equilibrium, we show that the non-rivalry of data fundamentally alters this picture: an exact Nash equilibrium (NE) need not exist, and in fact, $1.363$-approximate NE may also not exist under uniform pricing. We therefore investigate relaxed equilibrium notions. Allowing sellers to use beyond-uniform pricing---specifically, piecewise-linear convex pricing functions---guarantees approximate stability within a constant factor: there exists a pricing profile in which no seller can improve revenue by a factor of two by deviating to any uniform price (a 2-approximate NE against uniform deviations).
Finally, our simulations demonstrate fast convergence and empirical approximation guarantees that outperform the worst-case bound of 2.

\end{abstract}

\section{Introduction}
\label{sec:intro}

The rapid growth of data-driven decision-making and the widespread adoption of data-centric technologies have cemented data as one of the most valuable assets of the $21^{\text{st}}$ century. Falling storage costs, together with major advances in data mining, analytics, and machine learning, have dramatically increased both the usefulness and the economic value of data. According to \citet{AcumenDM}, the U.S.\ big data market is expected to reach approximately \$473~billion by 2030, highlighting the growing role of data in shaping innovation, productivity, and market dynamics.

A fundamental question in today's data economy is \emph{``how is data priced by data sellers?''}.
In this paper, we study data pricing in \emph{oligopolistic markets}, i.e., markets where a small number of sellers strategically choose prices for their datasets. These seller interactions give rise to a pricing game, and equilibrium prices correspond to a pure Nash equilibrium of this game.
This framework provides a natural parallel to data pricing in perfectly competitive markets—where sellers are price takers and prices are determined by equating supply and demand for each dataset—studied in recent work~\citep{ChaudhuryGMS26}.

\paragraph{Value of data.} To formalize our framework, we begin by specifying how data generates value for a buyer. In our model, buyers represent AI/ML agents who acquire data with the goal of improving the quality of their predictions. Formally, each buyer $i$ has a maximum budget of $b_i$, and seeks to infer an unknown parameter $\theta_i$, corresponding to an underlying quantity of interest such as future demand, user behavior, or system performance. Individual data records are modeled as digitized observations that provide noisy information about $\theta_i$. Data supplied by a given seller is drawn from a fixed distribution that reflects the seller's domain of activity and the informativeness of the data with respect to $\theta_i$. These distributions may differ across sellers, capturing heterogeneity in both data generation processes and predictive relevance. As an illustration, consider a buyer attempting to predict future demand for a product. One seller may provide historical transaction logs capturing purchase frequencies and timing, while another may offer mobility or web-traffic data that indirectly correlates with consumer interest. Although both datasets convey information about the same latent parameter, they differ in structure, noise characteristics, and informational content.

Suppose there are $m$ sellers. A buyer's acquisition decision is represented by a \emph{data bundle} $\mathbf{x}_i = (x_{i,1}, x_{i,2}, \ldots, x_{i,m})$, where $x_{i,j}$ denotes the number of data records purchased by buyer $i$ from seller $j$. Seller $j$'s dataset can be viewed as a set of $s_j$ data-records. The full collection of acquired signals, denoted by $S(\mathbf{x}_i)$, consists of one signal per data record. Upon observing these signals, the buyer updates her belief about the latent parameter, yielding a posterior distribution $\theta_i \mid S(\mathbf{x}_i)$. Following the recent literature on the economic value of data~\citep{baley2025data}, and the recent work on data pricing~\citep{ChaudhuryGMS26}, we define the buyer's \emph{utility} from $\mathbf{x}_i$, denoted by $u_i(\mathbf{x}_i)$, as the resulting reduction in uncertainty about $\theta_i$. Specifically, utility is measured as the expected increase in precision:
\[ u_i(\mathbf{x}_i) = \alpha_i \cdot (\mathbb{E}\!\left[\mathrm{Pre}(\theta_i \mid S(\mathbf{x}_i))\right]
- \mathrm{Pre}(\theta_i)), \]
where the precision of a random variable is defined as the inverse of its variance, i.e., $\mathrm{Pre}(\theta) = 1 / \mathrm{Var}(\theta)$, and $\alpha_i$ is the value per unit increase in precision for buyer $i$.

\paragraph{Pricing game and Nash equilibria.} Each seller $j$ chooses a price $p_j$ per data record of their dataset. This pricing model reflects common practices in real data marketplaces. For example, commercial platforms such as Snowflake Marketplace allow data providers to charge buyers on a usage-based basis, including per-query and per-row pricing for access to paid datasets~\citep{suger_snowflake_pricing}. Similarly, third-party data marketplaces employ volume-based pricing schemes, with providers such as Bright Data charging per thousand records accessed~\citep{brightdata2026datamarketplaces}.

Given a profile of seller prices $\bm p = (p_1, p_2, \dots, p_m)$, each buyer $i$ demands an affordable bundle $\bm x_i$
that maximizes her utility, i.e.,
\[ \bm x_i \in \argmax_{\bm z \in \mathbb{R}^m_{\geq 0}:\,\bm p^T \bm z \le b_i} u_i(\bm z). \]
When there are multiple optimal bundles, we use a natural tie-breaking rule (see \cref{sec:approx-ne} for details).
Given a price profile $\bm p$, let $\bm x^{*}_i (\bm p)$ denote the optimal demand bundle for buyer $i$.
Then, the revenue of seller $j$, $r_j(p_j, \bm p_{-j}) = p_j \cdot \sum_i x^*_{i,j}(\bm p)$.
A price profile $\bm p$ is an equilibrium if no seller can increase her revenue by unilaterally deviating;
that is, for every seller $j$ and every alternative price $p_j'$, we have $r_j(p_j, \bm p_{-j}) \geq r_j(p'_j, \bm p_{-j})$.

\paragraph{Economic effects of non-rivalry.} It is well known that in classical markets with \emph{rivalrous goods} and \emph{linear buyer utilities}\footnote{A good is rivalrous if its availability to one buyer is affected by its consumption by other buyers.}, any \emph{competitive equilibrium} (CE)—that is, a price profile at which each good's aggregate demand equals its supply—is also an oligopolistic equilibrium. In particular, at a CE no seller has an incentive to unilaterally deviate from her price (we provide a brief argument in \cref{sec:ce-is-ne-rival}). This equivalence, however, fails to extend to markets with \emph{non-rivalrous} goods, such as data. Indeed, one can show that the notion of competitive equilibrium introduced in~\citet{ChaudhuryGMS26} does not, in general, constitute a Nash equilibrium (NE) in oligopolistic data markets—standing in sharp contrast to the classical rivalrous setting. These observations underscore that classical CE theory cannot be directly applied to oligopolistic markets with non-rivalrous goods, necessitating a separate analysis of equilibrium in data markets, which is the focus of this paper.

\subsection{Our Contributions}

Consistent with~\citet{ChaudhuryGMS26}, for all our results, we assume that for each buyer $i$, $\theta_i \sim \mathcal N(0, \tau_i^{-1})$. Each data-record of seller $j$ is a signal $s_{i,j} = \theta_i + \eta_{i,j}$ to buyer $i$, where $\eta_{i,j} \sim  \mathcal N (0, \tau_{i,j}^{-1})$. It can be shown that under the foregoing assumptions, we have $u_i(\bm x_i) = \alpha_i \cdot \sum_j \tau_{i,j} x_{i,j}$. We now enlist our main contributions.

\begin{enumerate}[leftmargin=*]
    \item We show that there exist instances of oligopolistic data markets that admit no equilibrium (\cref{sec:lne-cex})---this is a sharp contrast to rivalrous oligopolistic markets. We also manage to find instances that admit no $1.363$-approximate Nash equilibrium (\cref{thm:lne-cex}).

\item We investigate natural relaxations of NE by allowing sellers to adopt more flexible pricing strategies. In particular, instead of assigning a uniform price to every data record, we allow sellers to use \emph{piecewise-linear convex} (PLC) pricing functions, which have been shown to be revenue-optimal for a monopolist selling data to heterogeneous buyers~\citep{ChaudhuryGSS26}. Unfortunately, even within the broader class of convex pricing strategies, we show that a pure NE may fail to exist (\cref{thm:convex-ne-cex} in \cref{sec:plc-cex}).

\item We show that computing an optimal PLC strategy (given other sellers' strategies) is computationally hard (\cref{sec:np-hardness-br}), suggesting that sellers may be expected to respond with coarse PLC pricing strategies or even uniform pricing, rather than optimizing over the full PLC space.
Motivated by this observation, we relax the equilibrium concept by restricting the class of deviations, rather than the strategy space itself. Specifically, we prove the existence of a PLC price profile such that no seller can deviate to a uniform pricing strategy and earn more than twice her current revenue, providing a meaningful approximate equilibrium despite the nonexistence of an exact one (\cref{thm:2-plc-ne-lindev}). This viewpoint parallels classical ideas in game theory and learning, where outcomes generated by rich or adaptive strategies are evaluated relative to a simpler benchmark class---most notably in regret minimization, where performance is compared against the best fixed strategy in hindsight. Here, uniform pricing serves as a natural behavioral baseline, grounding the relaxation in both theory and practice.

    \item Finally, we run some empirics to show the robustness of our results. We simulate a \emph{news data marketplace} where $n$ firms purchase data from $m$ news agencies to predict stock price movements, from a real-world dataset~\citep{aaron7sun_stocknews_2016}. Each firm evaluates the informativeness of a dataset by training a logistic regression model on historical news and measuring prediction error. Sellers update their pricing strategies via an \emph{inertial approximate best-response dynamics} (defined in \cref{sec:empirics}). Empirically, the dynamics converge efficiently across all settings, exhibiting a non-monotone pattern: convergence is slower when the number of sellers is very small due to strong competition effects, improves sharply as the market grows, and then increases steadily in larger markets reflecting a large-market stabilization effect. Secondly, despite theoretical guarantees being only for approximate equilibria, the dynamics reach an exact Nash equilibrium in all our simulations, highlighting a gap between worst-case bounds and typical behavior and motivating further theoretical investigation into structural properties of real-world pricing games. Overall, our empirical findings reveal several intriguing behaviors that merit more rigorous investigation in future work.
\end{enumerate}

\paragraph{Beyond Linear Utilities.}
Our results also work for slightly more general buyers' utility functions:
$u_i(\xvec_i) = g_i\big(\sum_j \tau_{i,j}x_{i,j}\big)$, where $g_i: \Rge \to \Rge$
is an increasing continuous function such that $g_i(0) = 0$.
This is because any strictly increasing transformation of a buyer's utility function
does not change her set of utility-maximizing bundles.
When $g_i$ is concave, this formulation captures diminishing marginal returns of data.

This more general form also captures other ways to model the value of data.
Instead of defining the value of data as the \emph{expected increase in precision},
we can also define it as the \emph{expected reduction in variance}, i.e.,
\[ u_i(\xvec_i) = \alpha_i \cdot (\Var(\theta_i) - \E[\Var(\theta_i \mid S(\xvec_i))]). \]
Assuming $\theta_i$ and $\eta_{i,j}$ to be normally-distributed as before, we get
\[ u_i(\xvec_i) = \alpha_i \cdot \left(\frac{1}{\tau_i} - \frac{1}{\tau_i + \sum_{j=1}^m \tau_{i,j}x_{i,j}}\right). \]
One can verify that this expression is concave in $\sum_{j=1}^m \tau_{i,j}x_{i,j}$.
Alternatively, we can define the value of data as the \emph{expected reduction in entropy}. This gives us
\[ u_i(\xvec_i) = \alpha_i \cdot \Bigg(\ln\bigg(\tau_i + \sum_{j=1}^m \tau_{i,j}x_{i,j}\bigg) - \ln(\tau_i)\Bigg). \]
This expression is also concave in $\sum_{j=1}^m \tau_{i,j}x_{i,j}$.

Since the choice of $g_i$ doesn't affect buyer behavior, we can let $g_i$ be the identity function \wLoG{}.
This makes $u_i$ linear, which simplifies exposition.

\subsection{Related Work}
\label{sec:related-work}

The economics of data has attracted growing attention due to the widespread adoption of AI and data-centric technologies. While a full overview is beyond the scope of this paper, we focus on the literature most closely related to our study. We adopt a general model of data value, where a buyer's utility is measured by the improvement in prediction accuracy. Prior work has explored more specialized valuation frameworks~\citep{farboodi2025valuing, veldkamp2023valuing, farboodi2023data}. We refer the reader to~\citet{fleckenstein2023review} for a detailed review of data valuation methods.

A significant line of research studies mechanisms that incentivize agents to share data, often by compensating for privacy loss~\citep{fallah2024optimal, cummings2023optimal, fallah2022bridging, MurhekarYCLM23}. There have been studies on mechanisms that incentivize sellers to truthfully report the variances of their datasets to a data aggregator, who aims to achieve a target prediction accuracy~\citep{cummings2015accuracy}.

Data markets, which match buyers' prediction requests with datasets from sellers, have been studied extensively in terms of strategic behavior, incentives, and revenue optimization. Monopolist revenue-maximization strategies have been analyzed~\citep{admati1986monopolistic, admati1990direct, bergemann2018design, BabaioffKP12}.
\citet{AgarwalDS19} designs a truthful mechanism where buyers pay in proportion to the accuracy gains they receive, yielding a discriminatory but explicit pricing scheme. Other work investigates equilibria and auctions in data markets with externalities~\citep{AgarwalDHR20, Hossain024}, first-principles approaches to data pricing~\citep{mehta2021sell, pei2020survey, cai2020sell, bergemann2022economics}, and stable outcomes in non-monetary data exchange economies~\citep{BhaskaraGIKMS24, akrami2025theoretical, song2025existencecomplexitycorestabledata}.

\section{Preliminaries}
\label{sec:prelims}

Let $[t]$ denote the set $\{1, 2, \ldots, t\}$ for any $t \in \{0\} \cup \mathbb{N}$.
For any $m \in \mathbb{N}$, define the simplex
$\Delta_m \defeq \{\xvec \in \mathbb{R}_{\ge 0}^m: \sum_{j=1}^m x_j = 1\}$.
For any $\xvec \in \mathbb{R}^m$, let $\supp(\xvec) \defeq \{j \in [m]: x_j \neq 0\}$.
For any $m \in \mathbb{N}$, let $\vecZero^{(m)}$ and $\vecOne^{(m)}$ denote
the $m$-dimensional vectors of all zeros and all ones, respectively.
When $m$ is clear from context, we simply write $\vecZero$ and $\vecOne$.
For any $j, m \in \mathbb{N}$, let $\vecE^{(j, m)}$ denote an $m$-dimensional vector
whose $j\Th$ component is 1 and all other components are 0.
When $m$ is clear from context, we write $\vecE^{(j)}$.

A data marketplace instance is given by the tuple $([n], [m], (u_i)_{i=1}^n, (b_i)_{i=1}^n)$.
Here $[n]$ is the set of buyers and $[m]$ is the set of sellers.
Each buyer $i$ has a budget $b_i \in \mathbb{R}_{\ge 0}$
and a utility function $u_i: [0, 1]^m \to \mathbb{R}_{\ge 0}$.
Given a bundle $\zvec = (z_1, \ldots, z_m)$, where $z_j$ is the fraction of dataset $j$,
buyer $i$ has a value of $u_i(\zvec)$ for it.

Based on our model of the value of data (see \cref{sec:intro}),
the utility function $u_i$ is \emph{linear}.
Specifically, each buyer $i$ has value $\tau_{i,j}$ for each dataset $j$,
and $u_i(\zvec) = \sum_{j=1}^m \tau_{i,j}z_j$.
(\Cref{sec:intro} uses an additional factor $\alpha_i$ in $u_i(\zvec)$.
We assume \wLoG{} that $\alpha_i = 1$ since we can scale all $\tau_{i,j}$ by $\alpha_i$ instead.)

There are many ways in which sellers can price their datasets.
A dataset having price $p$ is said to be priced uniformly/linearly if an $x$ fraction of the dataset costs $p \cdot x$.

\subsection{Buyer Behavior for Linear Pricing}
\label{sec:notation:revenue}

Suppose each seller $j \in [m]$ sets a linear price of $p_j$ for her dataset.
Each buyer $i \in [n]$ would like to purchase a bundle of datasets that maximizes her utility
subject to her budget constraint. If buyer $i$ purchases an $x_{i,j}$ fraction of each dataset $j$,
then she would like to maximize $\sum_{j=1}^m \tau_{i,j}x_{i,j}$ under the constraint
$\sum_{j=1}^m p_jx_{i,j} \le b_i$.
We also assume that the buyer only purchases datasets that give her positive value,
i.e., $x_{i,j} > 0$ for some dataset $j$ only if $\tau_{i,j} > 0$.
This is the fractional knapsack problem, whose solution is well-understood.
Define buyer $i$'s \emph{bang-per-buck} for dataset $j$ to be $\tau_{i,j}/p_j$.
To maximize utility, the buyer will first sort the datasets in non-increasing order of bang-per-buck.
Let $(\sigma_{i,1}, \ldots, \sigma_{i,m})$ be this ordering.
She would then buy as much of dataset $\sigma_{i,1}$ as possible, then buy as much of dataset $\sigma_{i,2}$
as possible, and so on, till she either exhausts her budget or she purchases every dataset of positive value.

A slight exception to the above behavior may arise when multiple datasets have the same bang-per-buck for a buyer.
For example, if $\sigma_{i,1}$ and $\sigma_{i,2}$ have the same bang-per-buck for buyer $i$,
and her budget is less than $p_{\sigma_{i,1}} + p_{\sigma_{i,2}}$,
then any way of distributing her budget across these two datasets is utility-maximizing behavior.
We assume that each buyer's \emph{tie-breaking rule}, i.e., how she distributes her budget
across multiple datasets having the same bang-per-buck, is known to all sellers.
This is necessary to ensure that the pricing game the sellers engage in is a perfect-information game.

\section{Non-Existence of (Approximate) NE}
\label{sec:lne-cex}

In this section, we show that there exist data market instances for which no (approximate) Nash equilibrium (NE) exists when sellers use linear pricing.

We begin by describing a family of data market instances.

\begin{example}
\label[example]{ex:lne-cex}
Let $\Ical$ be a data market instance with two sellers and $n$ buyers. There are two types of buyers: $n-1$ poor buyers and one rich buyer. Let $\alpha$ and $\beta$ be constants such that $1 < \alpha \le \beta$.
Each poor buyer has a budget of 1 and her value for the two datasets are
$\tau_{1,1} = \alpha$ and $\tau_{1,2} = 1$, respectively.
The rich buyer has a budget of $\beta$ and her values for the two datasets are
$\tau_{2,1} = \beta$ and $\tau_{2,2} = 0$, respectively.
\end{example}

Since the rich buyer values dataset 2 at zero, she will spend her budget only on dataset 1.
The poor buyers consider dataset 1 to be $\alpha$ times as valuable as dataset 2.
Thus, if dataset 1 costs more than $\alpha$ times dataset 2, they will prioritize purchasing 2. Similarly, if it costs less than $\alpha$ times dataset 2, they will prioritize 1. (For now, we assume that if dataset 1 costs exactly $\alpha$ times dataset 2,
buyers prefer 1. This assumption is relaxed in \cref{sec:lne-cex-extra}.)

The sellers engage in a pricing game with each other. A strategy profile for this game is given by $(p, q)$, where $p$ and $q$ are the prices of datasets 1 and 2, respectively.
We say that $(p, q)$ is a $c$-approximate NE for the pricing game
if no seller can increase her revenue by more than a factor of $c$ by changing her dataset's price. We show that, for a suitable choice of $\alpha$ and $\beta$,
a $1.363$-approximate NE does not exist.

\subsection{Preliminary Observations}
\label{sec:lne-cex:warmup}

First, let us build some intuition on why a $c$-approximate NE may not exist when $c$ is very close to 1 and $\alpha = \beta = n-1 \ge 3$.
Note that seller 1 has no incentive to price above $\beta$, since no buyer can pay more than $\beta$, and seller 2 has no incentive to price above 1, since the rich buyer is not interested and no poor buyer can pay more than 1. For ease of exposition, we therefore assume $p \le \beta$ and $q \le 1$.

\paragraph{Undercuts must be close.}
If $q < p/\alpha$, then seller 2's revenue is $(n-1)q$, which is increasing in $q$.
When $p \le \alpha q$, then seller 1's revenue is $p + (n-1)\min(p, 1)$, which is increasing in $p$.
Thus, whoever is undercutting the other, will do so using the maximum possible price.
Thus, if $(p, q)$ is a $c$-approximate NE, then $p \approx \alpha q$.

\paragraph{For large $p$, undercutting helps significantly.}
If $p \ge 1$ and $p \le \alpha q$, then decreasing dataset 2's price to
slightly less than $p/\alpha$ increases seller 2's revenue from 0 to $\approx (n-1)(p/\alpha)$.
If $p > \alpha q$, then decreasing dataset 1's price to slightly less than $\alpha q$
increases seller 1's revenue from $\approx \alpha q + (n-1)(1-q)$ to $\approx \alpha q + (n-1)$.
Thus, when $p \ge 1$, then some seller can always improve her revenue significantly
by undercutting the other, so $(p, q)$ is not a $c$-approximate NE.

\paragraph{For small $p$, seller 1 raises prices significantly.}
If $p \le 1$ (and so $q \le 1/\alpha$), then seller 1's revenue is at most $n$.
However, if she increases dataset 1's price to $\beta$, her revenue is at least
$\beta + (n-1)(1-q) \ge 2n-3$.
Thus, when $p$ is small, seller 1 can always improve her revenue significantly,
so $(p, q)$ is not a $c$-approximate NE.

Thus, for $c$ close to 1, $\alpha = \beta = n-1 \ge 3$, $p \le \beta$, and $q \le 1$,
$(p, q)$ is not a $c$-approximate NE.

\subsection{Stronger Inapproximability}
We now build on and formalize the ideas from \cref{sec:lne-cex:warmup} to obtain a stronger inapproximability result.

\begin{theorem}
\label{thm:lne-cex}
In the data market instance of \cref{ex:lne-cex},
if we set $\alpha = 0.733(n-1)$, $\beta = 0.860(n-1)$, and $n \to \infty$,
then a $1.363$-approximate Nash equilibrium does not exist.
\end{theorem}
\begin{proof}[Proof sketch]
Let $r_1(p, q)$ and $r_2(p, q)$ denote the revenue earned by sellers 1 and 2, respectively,
when seller 1 prices her dataset at $p$ and seller 2 prices her dataset at $q$.
Let $r_1^*(q)$ be the maximum revenue seller 1 can earn when seller 2 prices her dataset at $q$,
and $r_2^*(p)$ be the maximum revenue seller 2 can earn when seller 1 prices her dataset at $p$.
To show that a $c$-approximate NE does not exist,
we must show that for all $p$ and $q$, either $r_1^*(q) > c r_1(p, q)$ or $r_2^*(p) > c r_2(p, q)$. Equivalently, if we define
\[ \mu(p, q) \defeq \max\left(\frac{r_1^*(q)}{r_1(p, q)}, \frac{r_2^*(p)}{r_2(p, q)}\right), \]
then it suffices to prove that $\inf_{p,q} \mu(p, q) > c$.

In \cref{sec:lne-cex:rev}, we give closed-form expressions for $r_1(p, q)$, $r_2(p, q)$, $r_1^*(q)$, and $r_2^*(p)$.
In \cref{sec:lne-cex:constr-and-points}, we define a constant $c^*$ (that depends on $\alpha$, $\beta$, and $n$).
In Appendices \ref{sec:lne-cex:sp-props} to \ref{sec:lne-cex:case4}, we prove that $\inf_{p,q} \mu(p, q) = c^*$.
Thus, for any $\eps > 0$, a $(c^*-\eps)$-approximate NE does not exist.
For $\alpha = 0.733(n-1)$, $\beta = 0.860(n-1)$, and sufficiently large $n$,
we get $c^* > 1.363$, so a 1.363-approximate NE does not exist.
\end{proof}

\section{PLC Pricing and Non-Existence of NE}
\label{sec:plc-cex}

In \cref{sec:lne-cex}, we showed that an approximate NE may not exist when sellers use linear pricing.
Thus, we investigate a more flexible way to price datasets: \emph{piecewise-linear convex (PLC) pricing}.

\subsection{PLC Pricing and Sharding}
\label{sec:plc-pricing}

Let us first understand what it means to price a dataset non-linearly.
Let $p: [0, 1] \to \mathbb{R}_{\ge 0}$ be a continuous monotone function such that $p(0) = 0$.
A dataset is said to have \emph{pricing function} $p$
if a fraction $x \in [0, 1]$ of the dataset costs $p(x)$.

PLC pricing has been shown to be revenue-optimal for a monopolist
selling data to heterogeneous buyers~\citep{ChaudhuryGSS26}.
Moreover, PLC pricing has a very natural interpretation in a data market:
a seller can emulate PLC pricing by dividing a large dataset into multiple \emph{shards}
and pricing each shard linearly, as the following example illustrates.

\begin{example}
Consider a tabular dataset with 1 million rows.
The seller breaks the dataset into two shards:
the first 0.4 million rows form the first shard,
and the next 0.6 million rows form the second shard.

Suppose the seller prices the first shard at \$5 per million rows,
and the second shard at \$8 per million rows.
Any utility-maximizing buyer will prefer the first shard, since it is cheaper per row.
If she purchases $x$ million rows in total, the total cost would be
\[ p(x) = \begin{cases}5x & \text{ if } x \le 0.4 \\ 5 \cdot 0.4 + 8(x - 0.4) & \text{ if } x > 0.4\end{cases}. \]
Note that $p$ is PLC.
\end{example}

\subsection{Non-Existence of NE}
\label{sec:convex-ne-cex}

We show that even under the richer class of PLC pricing functions and just two sellers, a Nash equilibrium may not exist.
In fact, this negative result also extends to the slightly more general class of convex pricing functions.
Intuitively, any convex function can be approximated arbitrarily closely by a PLC function
with a sufficiently large number of shards, so one would expect these function classes to behave similarly.

\begin{restatable}{theorem}{thmConvexNeCex}
\label{thm:convex-ne-cex}
Consider a data market instance with $n \ge 2$ buyers and 2 datasets.
Each buyer has value 1 for each dataset.
The first $n-1$ buyers are called \emph{poor} and have a budget of 1 each.
The $n\Th$ buyer is called \emph{rich} and has a budget of $n+1$.
In this instance, a Nash equilibrium does not exist for convex pricing functions.
\end{restatable}
\begin{proof}[Proof sketch]
The total budget among the buyers is $(n-1) + (n+1) = 2n$.
For any pricing function selected by a seller, we show that the other seller
can pick a pricing function that gives her revenue more than $n$.
This would prove that a Nash equilibrium doesn't exist,
because in every pair of pricing functions, some seller will always earn at most $n$,
and that seller can profitably deviate to earn more than $n$.

For all $j \in [2]$, let $r_j(p_1, p_2)$ denote seller $j$'s revenue for pricing profile $(p_1, p_2)$.
By symmetry, we only need to prove one side: given seller 1's pricing function $p_1$,
we need to find a pricing function $p_2$ such that $r_2(p_1, p_2) > n$.
We give a sketch of how to construct such a $p_2$.
The full proof can be found in \cref{sec:plc-cex-extra}.

\paragraph{We can extract every buyer's full budget.}
Suppose $r_1(p_1, p_1) + r_2(p_1, p_1) < 2n$, i.e., some buyers are not exhausting their budgets.
We show that seller 2 can pick another pricing function $p_2$ which nearly preserves her revenue
from all buyers, and additionally redirects unspent budget to her.
Specifically, $r_2(p_1, p_2)$ can be made arbitrarily close to $2n - r_1(p_1, p_1)$,
so if $r_1(p_1, p_1) < n$, then we can get $r_2(p_1, p_2) > n$.

For a sufficiently small $\delta > 0$,
let $p_2'(x)$ be $p_1'(x/(1-\delta))$ if $x < 1-\delta$
and $p_2'(x)$ be a very large constant if $x > 1-\delta$ (see \cref{fig:budget-extract}).
(We use the notation $f'(x)$ to denote the sub-derivative of $f$ at $x$.
See \cref{sec:plc-cex-extra} for details.)
Every buyer's expenditure on dataset 2 at price $p_2$ is at least
$1-\delta$ times the expenditure at price $p_1$.
If a buyer had an unspent budget of $\mu > 0$ at profile $(p_1, p_1)$, then she was purchasing both datasets completely,
so her expenditure on dataset 1 cannot increase when dataset 2's price changes to $p_2$.
Since $p_1(1) + p_2(1) \ge n+1$, her expenditure on dataset 2 must increase by $\mu$.
Hence, for $\delta \to 0$, we get $r_2(p_1, p_2) \to 2n - r_1(p_1, p_1)$.

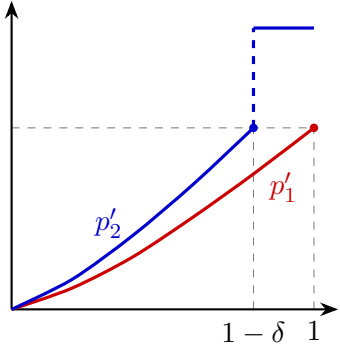
\begin{figure}[htb]
\centering
\begin{tikzpicture}[x=4cm,y=1.2cm,>={Stealth[length=2.4mm]}]

  \path
    plot[smooth] coordinates {(0,0) (0.2,0.330) (0.4,0.812) (0.6,1.376) (0.8,2)}
    -- plot[smooth] coordinates {(0.8,1.492) (0.6,1.016) (0.4,0.576) (0.2,0.234) (0,0)}
    -- cycle;

  \draw[gray,dashed] (0,2) -- (1,2);      %
  \draw[gray,dashed] (0.8,0) -- (0.8,2);    %
  \draw[gray,dashed] (1,0) -- (1,2);        %

  \draw[->,thick] (0,0) -- (0,3.4);
  \draw[->,thick] (0,0) -- (1.08,0);

  \draw[red!80!textColor,very thick]
    plot[smooth] coordinates {(0,0) (0.2,0.234) (0.4,0.576) (0.6,1.016) (0.8,1.492) (1,2)};

  \draw[blue!80!textColor,very thick]
    plot[smooth] coordinates {(0,0) (0.2,0.330) (0.4,0.812) (0.6,1.376) (0.8,2)};
  \draw[blue!80!textColor,very thick,dashed] (0.8,2) -- (0.8,3.1);  %
  \draw[blue!80!textColor,very thick] (0.8,3.1) -- (1,3.1);         %

  \fill[blue!80!textColor] (0.8,2) circle (1.6pt);
  \fill[red!80!textColor]  (1,2)   circle (1.6pt);

  \node[blue!80!textColor,anchor=east] at (0.40,0.95) {$p_2'$};
  \node[red!80!textColor,anchor=west]  at (0.82,1.35) {$p_1'$};

  \node[below] at (0.8,-0.02) {$1-\delta$};
  \node[below] at (1,-0.02)   {$1$};
\end{tikzpicture}

\caption[Budget extraction]{\normalfont
The derivatives $p_1'$ (red) and $p_2'$ (blue),
where $p_2'(x) = p_1'(x/(1-\delta))$ for $x \in (0, 1-\delta)$
is a horizontal compression of $p_1'$, and $p_2'(x)$ is a very large
constant for $x \in (1-\delta, 1)$.}
\label{fig:budget-extract}
\end{figure}

Now assume $r_1(p_1, p_1) = r_2(p_1, p_1) = n$.

\paragraph{Flattening prices.}
Since $p_1$ is monotone and convex, we get that $p_1(x) \le x \cdot p_1'(x)$ for all $x \in [0, 1]$,
and $p_1(x) + x \cdot p_1'(x)$ is non-decreasing in $x$.
Thus, pick $\alpha \in (0, 1)$ such that $p_1(\alpha) + \alpha \cdot p_1'(\alpha) = 1$.
(Such an $\alpha$ always exists; see \cref{sec:plc-cex-extra}.)
Define $p_2'(x) \defeq p_1'(\alpha)$ for all $x \in [0, \alpha]$.
Then every buyer will spend the first 1 dollar of their budget as follows:
$p_1(\alpha)$ on dataset 1 and $p_2(\alpha) = \alpha \cdot p_1'(\alpha)$ on dataset 2.
Note that $p_2(\alpha) \ge p_1(\alpha)$, and the inequality is strict if
$p_1'(x) \neq p_1'(\alpha)$ for some $x \in (0, \alpha)$.

\begin{figure}[htb]
\centering
\begin{tikzpicture}[x=5cm,y=1.4cm,>={Stealth[length=2.4mm]}]

  \fill[pattern=north east lines,pattern color=red!15!bgColor]
    (0,0) .. controls (0.02,0.058) and (0.078,0.208) .. (0.12,0.35)
          .. controls (0.162,0.492) and (0.203,0.675) .. (0.25,0.85)
          .. controls (0.297,1.025) and (0.35,1.217) .. (0.4,1.4)
    -- (0.4,0) -- cycle;
  \fill[pattern=north east lines,pattern color=red!15]
    (0.4,1.4) .. controls (0.45,1.583) and (0.5,1.78) .. (0.55,1.95)
              .. controls (0.6,2.12) and (0.65,2.278) .. (0.7,2.42)
              .. controls (0.75,2.562) and (0.813,2.703) .. (0.85,2.8)
    -- (0.85,0) -- (0.4,0) -- cycle;
  \fill[pattern=north west lines,pattern color=blue!15!bgColor]
    (0,1.4) -- (0.4,1.4)
    .. controls (0.35,1.217) and (0.297,1.025) .. (0.25,0.85)
    .. controls (0.203,0.675) and (0.162,0.492) .. (0.12,0.35)
    .. controls (0.078,0.208) and (0.02,0.058) .. (0,0) -- cycle;
  \fill[pattern=north west lines,pattern color=blue!15!bgColor]
    (0.4,2.8) -- (0.85,2.8)
    .. controls (0.813,2.703) and (0.75,2.562) .. (0.7,2.42)
    .. controls (0.65,2.278) and (0.6,2.12) .. (0.55,1.95)
    .. controls (0.5,1.78) and (0.45,1.583) .. (0.4,1.4) -- cycle;

  \draw[gray,dashed] (0.4,0) -- (0.4,1.4);
  \draw[gray,dashed] (0.85,0) -- (0.85,2.8);

  \draw[->,thick] (0,0) -- (0,3.15);
  \draw[->,thick] (0,0) -- (0.95,0);

  \draw[blue!80!textColor,very thick] (0,1.4) -- (0.4,1.4);          %
  \draw[blue!80!textColor,very thick] (0.4,2.8) -- (0.85,2.8);       %
  \draw[blue!80!textColor,very thick,dashed] (0.4,1.4) -- (0.4,2.8); %

  \draw[red!80!textColor,very thick]
    (0,0) .. controls (0.02,0.058) and (0.078,0.208) .. (0.12,0.35)
          .. controls (0.162,0.492) and (0.203,0.675) .. (0.25,0.85)
          .. controls (0.297,1.025) and (0.35,1.217) .. (0.4,1.4)
          .. controls (0.45,1.583) and (0.5,1.78) .. (0.55,1.95)
          .. controls (0.6,2.12) and (0.65,2.278) .. (0.7,2.42)
          .. controls (0.75,2.562) and (0.813,2.703) .. (0.85,2.8);

  \fill[red!80!textColor] (0.4,1.4) circle (1.6pt);
  \fill[red!80!textColor] (0.85,2.8) circle (1.6pt);

  \node at (0.25,0.42) {$A$};
  \node at (0.10,0.95) {$B$};
  \node at (0.62,0.85) {$C$};
  \node at (0.55,2.42) {$D$};

  \node[blue!80!textColor,anchor=west] at (0.53,3.00) {$p_2'$};
  \node[red!80!textColor,anchor=west]  at (0.60,2.05) {$p_1'$};

  \node[below]      at (0.4,-0.02) {$\alpha$};
  \node[below]      at (0.85,-0.02){$\beta$};
\end{tikzpicture}

\caption[Price flattening]{\normalfont
The derivatives $p_1'$ (red) and $p_2'$ (blue),
where $p_2'(x) = p_1'(\alpha)$ for $x \in (0, \alpha)$
and $p_2'(x) = p_1'(\beta)$ for $x \in (\alpha, \beta)$.
The thresholds $\alpha$ and $\beta$ are defined such that
the shaded areas satisfy $2A+B = 1$ and $2C+D=n$.
We get $r_2(p_1, p_2) = r_1(p_1, p_1) + n \cdot C + D$.}
\label{fig:price-flatten}
\end{figure}
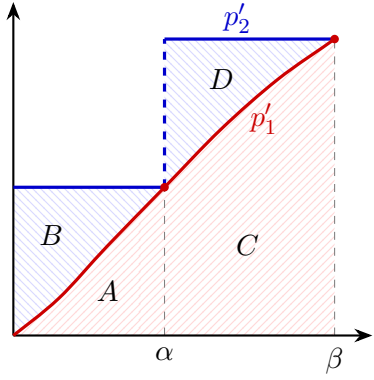

Similarly, we can find $\beta \in (\alpha, 1]$ such that
$p_1(\beta) - p_1(\alpha) + (\beta - \alpha)p_1'(\beta) = n$.
Let $p_2'(x) \defeq p_1'(\beta)$ for all $x \in (\alpha, \beta]$.
Then the rich buyer will spend the remaining $n$ dollars of her budget as follows:
$p_1(\beta) - p_1(\alpha)$ on dataset 1 and $(\beta - \alpha)p_1'(\beta)$ on dataset 2.
The latter is at least the former, and the inequality is strict if
$p_1'(x) \neq p_1'(\beta)$ for some $x \in (\alpha, \beta)$.
See \cref{fig:price-flatten} for an example of $p_1$ and $p_2$.

Thus, using a 2-shard pricing function, seller 2 earns at least as much as seller 1.
Moreover, if $p_1'(x) \neq p_2'(x)$ for some $x \in (0, \beta)$,
then the inequality is strict, and we get $r_2(p_1, p_2) > n$.

\paragraph{Large uniform price to defeat 2-shard PLC pricing.}
The only remaining case involves a highly structured $p_1$:
it has a shard of length $\alpha$ and price-per-unit $p_1'(\alpha)$,
and a shard of length $\beta - \alpha$ and price-per-unit $p_1'(\beta)$.
The price of shard 1 is $\alpha \cdot p_1'(\alpha) = 1/2$,
and the price of shard 2 is $(\beta - \alpha)p_1'(\beta) = n/2$.

We price dataset 2 at $p_2(x) = p_1'(\beta)(1-\eps)x$, where $\eps > 0$ is infinitesimally small.
If $p_1'(\alpha) = p_1'(\beta)$, then $r_2(p_1, p_2) = (n-1) + (n+1)(1-\eps)/2 > n$.
Otherwise, $r_2(p_1, p_2) = (n-1)/2 + \min(n+\frac{1}{2}, p_1'(\beta)(1-\eps))$.
Let $\nu \defeq \alpha \cdot (p_1'(\beta) - p_1'(\alpha)) > 0$. Then
\[ p_1'(\beta) \ge \beta \cdot p_1'(\beta)
= \nu + \alpha \cdot p_1'(\alpha) + (\beta - \alpha) p_1'(\beta) = \nu + (n+1)/2, \]
Thus, $r_2(p_1, p_2) > n$.

To summarize, we showed that for any $p_1$, there exists $p_2$ such that $r_2(p_1, p_2) > n$.
Thus, in any pricing profile, the seller earning less than $n$ can deviate to earn more than $n$,
so a Nash equilibrium does not exist.
\end{proof}

\section{Approximate NE for PLC Pricing}
\label{sec:approx-ne}

In \cref{sec:plc-cex}, we showed that an NE may not exist when sellers use PLC pricing.

At a Nash equilibrium under PLC pricing, no seller can increase her revenue by switching to a different PLC pricing strategy.
However, we show that computing a seller's best response under PLC pricing is NP-hard (\cref{sec:np-hardness-br}).
Consequently, even if a pricing profile is not a Nash equilibrium, sellers may find it difficult to deviate.

Thus, we ask whether there exist pricing profiles that are at least stable against deviations to linear prices.
This guarantee is reminiscent of results in game theory, particularly regret minimization,
where an adaptive strategy---more general than a fixed strategy---is compared to a fixed strategy in hindsight.

\paragraph{Discretizing prices.}
The strategy space of all PLC functions is difficult to work with.
For ease of exposition, we simplify the strategy space by \emph{discretizing} prices.
Formally, for each seller $j \in [m]$, we are given a finite non-empty set $P_j \subset \mathbb{R}_{> 0}$,
and we require the price of each shard of dataset $j$ to lie in $P_j$.
This is a mild restriction, since any PLC function $p_j$ can be approximated
arbitrarily closely by choosing a sufficiently large $P_j$.
Moreover, in many real-world settings, prices are typically expressed as
integer multiples of a base currency, e.g., $P_j = \{1, 2, \ldots, 10^6\}$,
meaning that each shard's price-per-unit must be a whole number of dollars and cannot exceed one million dollars.

The ideas in \cref{sec:lne-cex:warmup} easily carry over to discretized linear prices,
provided that the set of prices is sufficiently fine-grained.
Thus, even with discretized prices, a Nash equilibrium in linear strategies may not exist.
Moreover, in \cref{sec:disc-plc-ne-cex}, we show that a Nash equilibrium may not exist
even for discretized PLC pricing.

Despite these negative results, we show that there exists a strategy profile for sellers under discretized PLC pricing
in which no seller can increase her revenue by more than a factor of two by deviating to a linear pricing strategy.

\paragraph{Strategy profile.}
Let $P_j = \{p_{j,1}, \ldots, p_{j,m_j}\}$ for each seller $j$, where $0 < p_{j,1} < \ldots < p_{j,m_j}$.
Seller $j$'s pricing strategy can be described in terms of the shard lengths associated with each price.
Formally, her strategy is given by the vector
$\ellvec_j = (\ell_{j,1}, \ldots, \ell_{j,m_j}) \in \Delta_{m_j}$,
where $\ell_{j,k}$ is the size of the $k\Th$ shard, which has price $p_{j,k}$.
Thus, if a buyer purchases a fraction $x$ of the dataset,
she would pay $p_{j,1}x$ if $x \le \ell_{j,1}$,
pay $p_{j,1}\ell_{j,1} + p_{j,2}(x-\ell_{j,1})$ if $0 \le x - \ell_{j,1} \le \ell_{j,2}$, and so on.

The strategy profile of all sellers is denoted by $\ellvec \defeq (\ellvec_1, \ldots, \ellvec_m)$,
and we write $\ell_{-j} \defeq (\ell_{j'})_{j' \neq j}$ for the strategies of all sellers other than $j$.

\paragraph{Buyer behavior and revenue.}
Once prices are fixed, each buyer chooses a bundle that maximizes her utility.
For linear prices, we showed in \cref{sec:notation:revenue} that each buyer faces a fractional knapsack problem,
which she solves by ordering the datasets in non-increasing order of bang-per-buck and purchasing greedily.
Similarly, under PLC pricing, buyers order the \emph{shards} in non-increasing order of bang-per-buck, and purchase greedily.
Let $r_j(\ellvec)$ denote the revenue earned by seller $j$ for the strategy profile $\ellvec$.
For each buyer $i \in [n]$, let $\sigma_i$ denote her ordering of shards by bang-per-buck,
i.e., the $k\Th$ shard of dataset $j$ has bang-per-buck $\tau_{i,j}/p_{j,k}$,
and $\sigma_i$ contains all pairs $(j, k)$ in non-increasing order of bang-per-buck.

We assume that each buyer $i$ uses \emph{sequential tie-breaking}:
when multiple shards have the same bang-per-buck, she preferentially allocates her budget
to the shard that appears first in $\sigma_i$
(as opposed to, say, purchasing multiple shards partially).
This assumption makes analysis much simpler.
This is a mild assumption because prices are discretized,
so a slight perturbation of utilities would eliminate ties in bang-per-buck.

We now state the main result of this section.

\begin{theorem}[2-NE for PLC pricing and linear deviations]
\label{thm:2-plc-ne-lindev}
There exists a strategy profile $\ellvec^*$ such that
no seller can increase her revenue by more than a factor of 2 by deviating to a linear pricing strategy.
Formally, for each seller $j \in [m]$, we have
\[ \max_{k \in [m_j]} r_j(\vecE^{(k)}, \ellvec^*_{-j}) \le 2r_j(\ellvec^*). \]
\end{theorem}

We prove \cref{thm:2-plc-ne-lindev} in two steps.
First, we define each seller $j$'s \emph{randomized revenue} $\rhat_j(\cdot)$
and show that there exists a strategy profile $\ellvec^*$ such that
no seller can increase her randomized revenue by deviating to a linear pricing.
Second, we show that for any strategy profile $\ellvec$ and any seller $j$,
the actual revenue is always at least half of the randomized revenue,
i.e., $r_j(\ellvec) \ge \rhat_j(\ellvec)/2$.

\paragraph{Randomized revenue.}
Consider the hypothetical scenario where seller $j$ prices the dataset linearly,
but the price is decided uniformly randomly,
whereas the remaining sellers use (deterministic) PLC pricing.
Specifically, seller $j$ sets the price of her dataset to $p_{j,k}$
with probability $\ell_{j,k}$ for all $k \in [m_j]$,
whereas every other seller $j'$ has a shard of size $\ell_{j',k}$
having price $p_{j',k}$ for all $k \in [m_{j'}]$.
Denote seller $j$'s expected revenue by $\rhat_j(\ellvec)$. Then,
\[ \rhat_j(\ellvec) = \sum_{k=1}^{m_j} \ell_{j,k} r_j(\vecE^{(k)}, \ell_{-j}). \]
Observe that when $\ellvec_j = \vecE^{(k)}$, then $r_j(\ellvec) = \rhat_j(\ellvec)$.

Under randomized revenues, we can show that a Nash equilibrium always exists
using Brouwer's fixed-point theorem.

\begin{lemma}[Brouwer's fixed-point \citep{border1985fixed}]
\label{thm:brouwer}
Let $S \subset \mathbb{R}^d$ be a convex, closed, and bounded set.
Let $f: S \to S$ be a continuous function.
Then $f$ has a fixed point, i.e., $\exists x^* \in S$ such that $f(x^*) = x^*$.
\end{lemma}

\begin{lemma}
\label{thm:semi-rand-ne}
There is a strategy profile $\ellvec^*$ such that no seller $j$ can
increase her randomized revenue by deviating to a linear pricing.
Formally, for all $j \in [m]$, we have
\[ \rhat_j(\ellvec^*) = \max_{k \in [m_j]} r_j(\vecE^{(k)}, \ellvec^*_{-j}). \]
\end{lemma}
\begin{proof}
We will define a function $f$ that perturbs a strategy profile in the direction of each seller's \emph{better response},
and show that $f$'s fixed point is the $\ellvec^*$ we need.

For any $j \in [m]$, $k \in [m_j]$, and any strategy profile $\ellvec$, let
\[ g_j(k, \ellvec) \defeq \max(0, r_j(\vecE^{(k)}, \ell_{-j}) - \rhat_j(\ellvec)). \]
Here, $g_j(k, \ellvec)$ represents the gain (relative to $\rhat_j$) that seller $j$ achieves by deviating from her current strategy $\ellvec_j$ to $\vecE^{(k)}$.
For any strategy profile $\ellvec$, define
\[ f(\ellvec)_{j,k} = \frac{\ell_{j,k} + g_j(k, \ellvec)}{1 + \sum_{k'=1}^{m_j} g_j(k', \ellvec)}. \]
Let $\Delta^* \defeq \prod_{j=1}^m \Delta_{m_j}$. Then, $f: \Delta^* \to \Delta^*$,
since for all $j \in [m]$, we have $\sum_{k=1}^{m_j} f(\ellvec)_{j,k} = 1$.

Since $\Delta^*$ is the product of simplices, it is closed and convex. It can be shown that for all $j \in [m]$, $r_j(\ellvec)$ and $\rhat_j(\ellvec)$ are continuous in $\ellvec$.
Hence, $g_j(k, \cdot)$ is continuous, and so, $f$ is continuous.
Hence, $f$ satisfies the conditions for applying Brouwer's fixed point theorem,
and by \cref{thm:brouwer}, $f$ has a fixed point $\ellvec^*$.

There exists $\ktild \in [m_j]$ such that $g_j(\ktild, \ellvec^*) = 0$, since
\begin{align*}
\rhat_j(\ellvec^*) &= \sum_{k \in \supp(\ellvec^*_j)} \ell^*_{j,k}r_j(\vecE^{(k)}, \ellvec^*_{-j})
    \ifTwoCol{\\ &}{\;\;}\ge\; \min_{k \in \supp(\ellvec^*_j)} r_j(\vecE^{(k)}, \ellvec^*_{-j}).
\end{align*}
Let $\alpha \defeq \sum_{k=1}^{m_j} g_j(k, \ellvec^*)$. If $\alpha > 0$, then
\[ \ell^*_{j,\ktild} = f(\ellvec^*)_{j,\ktild} = \frac{\ell^*_{j,\ktild}}{1 + \alpha} < \ell^*_{j,\ktild}, \]
which is a contradiction. Hence, $\alpha = 0$, and so, $g_j(k, \ellvec^*) = 0$ for all $k \in [m_j]$.
Thus, for all $j \in [m]$ and $k \in [m_j]$, we have
$r_j(\vecE^{(k)}, \ellvec^*_{-j}) \le \rhat_j(\ellvec^*)$.
Moreover, we also have
\[ \rhat_j(\ellvec^*) = \sum_{k=1}^{m_j} \ell^*_{j,k} r_j(\vecE^{(k)}, \ellvec^*_{-j})
    \le \max_{k=1}^{m_j} r_j(\vecE^{(k)}, \ellvec^*_{-j}). \]
Hence, for all $j \in [m]$, we have
\[ \rhat_j(\ellvec^*) = \max_{k=1}^{m_j} r_j(\vecE^{(k)}, \ellvec^*_{-j}).
\qedhere \]
\end{proof}

\paragraph{From randomized to PLC revenue.}
We now explore the relationship between PLC and randomized revenue.

\begin{lemma}
\label{thm:randrev-vs-rev}
For any seller $j \in [m]$ and any strategy profile $\ellvec$, we have $r_j(\ellvec) \ge \rhat_j(\ellvec)/2$.
\end{lemma}
\begin{proof}
Fix a seller $j \in [m]$ and the strategy profile $\ellvec$.
We now express the revenue and randomized revenue in terms of \emph{residual budgets}.
For any buyer $i \in [n]$, recall that $\sigma_i$ represents the sequence of all shards
of all datasets, ordered in non-increasing order of bang-per-buck.

Let $L_{i,k}$ be the set of all shards of all datasets in $[m] \setminus \{j\}$
that precede shard $k$ of dataset $j$ in $\sigma_i$.
Let $L'_{i,k}$ be the set of all shards of all datasets that precede shard $k$ of dataset $j$ in $\sigma_i$.
Define $\gamma_{i,k}$ as buyer $i$'s remaining budget after purchasing everything in $L_{i,k}$,
and $\gamma'_{i,k}$ as $i$'s remaining budget after purchasing everything in $L'_{i,k}$.
Formally,
\begin{align*}
\gamma_{i,k} &\defeq \max\left(0, b_i - \sum_{(j', k') \in L_{i,k}} p_{j',k'}\ell_{j',k'} \right),
\ifTwoCol{\\}{&} \gamma'_{i,k} &\defeq \max\left(0, b_i - \sum_{(j', k') \in L'_{i,k}} p_{j',k'}\ell_{j',k'} \right).
\end{align*}
Then, we have
\begin{align*}
\rhat_j(\ellvec) &= \sum_{k=1}^{m_j} \sum_{i=1}^n \ell_{j,k}\min(p_{j,k}, \gamma_{i,k}),
\ifTwoCol{\\}{&} r_j(\ellvec) &= \sum_{k=1}^{m_j} \sum_{i=1}^n \min(p_{j,k}\ell_{j,k}, \gamma'_{i,k}).
\end{align*}

If seller $j$ prices her dataset linearly at $p_{j,k}$ with probability $\ell_{j,k}$ for each $k \in [m_j]$,
then the expected revenue from buyer $i$ is $\rhohat_i = \sum_{k=1}^{m_j} \shat_{i,k}$,
where $\shat_{i,k} \defeq \ell_{j,k} \min(p_{j,k}, \gamma_{i,k})$.

On the other hand, if seller $j$ prices her dataset deterministically, where the shard of price $p_{j,k}$ has size $\ell_{j,k}$,
then the revenue from buyer $i$ and shard $k$ is $s_{i,k} \defeq \min(p_{j,k}\ell_{j,k}, \gamma'_{i,k})$.
The total revenue from buyer $i$ is $\rho_i = \sum_{k=1}^{m_j} s_{i,k}$.
We will show that $2\rho_i \ge \rhohat_i$ for all $i \in [n]$, which implies that
\[ 2r_j(\ellvec) = \sum_{i=1}^n (2\rho_i) \ge \sum_{i=1}^n \rhohat_i = \rhat_j(\ellvec). \]

Suppose buyer $i$ buys $n_i$ shards fully, i.e., $s_{i,k} = p_{j,k}\ell_{j,k}$ for all $k \in [n_i]$. Then
\[ \rho_i \ge \sum_{k=1}^{n_i} s_{i,k} = \sum_{k=1}^{n_i} p_{j,k}\ell_{j,k}
    \ge \sum_{k=1}^{n_i} \shat_{i,k}. \]
If $n_i = m_j$, then $\rho_i \ge \rhohat_i$, and we are done.
Now let $n_i < m_j$. Then $\ell_{j,n_i+1} > 0$ and
\begin{align*}
s_{i,n_i+1} = \gamma'_{i,n_i+1} & = \max\left(0, \gamma_{i,n_i+1} - \sum_{k=1}^{n_i} p_{j,k}\ell_{j,k}\right)
\ifTwoCol{\\ &}{} < p_{j,n_i+1}\ell_{j,n_i+1}.
\end{align*}
There are two cases, depending on whether buyer $i$ runs out of money before reaching shard $n_i+1$.
If $s_{i,n_i+1} = 0$, i.e., buyer $i$ ran out of money before reaching shard $n_i+1$, then
\[ \gamma_{i,n_i+1} \le \sum_{k=1}^{n_i} p_{j,k}\ell_{j,k} = \sum_{k=1}^{n_i} s_{i,k} \le \rho_i. \]
If $s_{i,n_i+1} > 0$, then buyer $i$ exhausts her budget on shard $n_i+1$, so $\rho_i = \gamma_{i,n_i+1}$.
In either case, we get $\rho_i \ge \gamma_{i,n_i+1}$.

Moreover, since $\gamma_{i,k}$ is non-increasing in $k$, we get
\begin{align*}
\gamma_{i,n_i+1} \ge \sum_{k=n_i+1}^{m_j} \ell_{j,k}\gamma_{i,n_i+1} &\ge \sum_{k=n_i+1}^{m_j} \ell_{j,k}\gamma_{i,k}
\ifTwoCol{\\ &}{} \ge \sum_{k=n_i+1}^{m_j} \shat_{i,k}.
\end{align*}
Thus,
\[ 2\rho_i = \rho_i + \rho_i
    \ge \left(\sum_{k=1}^{n_i} \shat_{i,k}\right) + \left(\sum_{k=n_i+1}^{m_j} \shat_{i,k}\right)
    = \rhohat_i. \]
Hence, $2r_j(\ellvec) \ge \rhat_j(\ellvec)$.
\end{proof}

By combining \cref{thm:semi-rand-ne,thm:randrev-vs-rev}, we get the proof of \cref{thm:2-plc-ne-lindev}.

\subsection{Sparsity of PLC Strategies}
\label{sec:no-of-shards}

A PLC strategy is said to be \emph{sparse} if it has only a few shards.
Specifically, seller $j$'s pricing strategy $\ell_j$ is said to be sparse if $|\supp(\ell_j)|$ is small.
Sparse strategies are nice because they are easy to represent and communicate.

From the proof of \cref{thm:2-plc-ne-lindev}, we can infer some structural properties
of the 2-approximate Nash equilibrium $\ell^*$ for linear deviations.
For any seller $j$, let $B_j$ be the indices of prices
that are linear best-responses, i.e.,
\[ B_j \defeq \argmax_{k=1}^{m_j} \; r_j(\vecE^{(k)}, \ell^*_{-j}). \]
Then by \cref{thm:semi-rand-ne}, we get $\supp(\ell^*_j) \subseteq B_j$.
If $B_j$ is small for all $j$, then each seller's pricing strategy is \emph{sparse},
i.e., involving only a few shards.
We leave proving an upper bound on $|B_j|$ as an interesting open problem.

We also consider a related problem: the sparsity of PLC best-responses.
Specifically, in \cref{sec:no-of-shards-extra}, we show that for any strategy profile $\ell$
and any seller $j$, there always exists a PLC best-response $\ellhat_j$
such that $|\supp(\ellhat_j)| \le n$, where $n$ is the number of buyers.

\section{Empirics}
\label{sec:empirics}

In this section, we investigate how to find an approximate Nash equilibrium on a \emph{news data market}
constructed from a real-world dataset~\citep{aaron7sun_stocknews_2016}.

\paragraph{Empirical Setup: Datasets and Buyer Valuations.}
The dataset~\citep{aaron7sun_stocknews_2016} consists of news headlines paired with indicators of subsequent stock price movements over time. To mimic the interaction between data sellers and buyers, we partition the dataset into two disjoint subsets: one used to simulate the datasets offered by sellers, and another used to simulate the buyer-side test data used for evaluation. In each simulation run, we first fix the number of buyers and sellers. The seller-side dataset is then distributed uniformly at random across sellers, creating heterogeneous but comparable datasets for different news agencies.

Each buyer is interested in predicting stock price movements on a specific subset of days, corresponding to the days on which the buyer may trade. This induces a buyer-specific prediction task and defines a corresponding test dataset consisting of news articles and stock price indicators for those days drawn from the buyer-side dataset that was originally partitioned for evaluation. To assess the relevance of a seller’s dataset for a given buyer, buyer $i$ trains a predictive model $f_{i,j} : \mathbf{x} \mapsto \{0,1\}$ using only the news headlines provided by seller $j$. The model is implemented as a logistic regression classifier based on term-frequency features extracted from the headlines. The trained model is then evaluated on buyer $i$’s test data, and the variance of its prediction errors is measured. This empirical variance serves as a proxy for the informativeness of seller $j$’s data for buyer $i$’s prediction task, and is used to infer $\tau_{i,j}$, the relevance of seller $j$’s dataset to buyer $i$.

\begin{figure*}[t]
    \centering
    \begin{subfigure}[t]{0.43\textwidth}
        \centering
        \includegraphics[width=\linewidth]{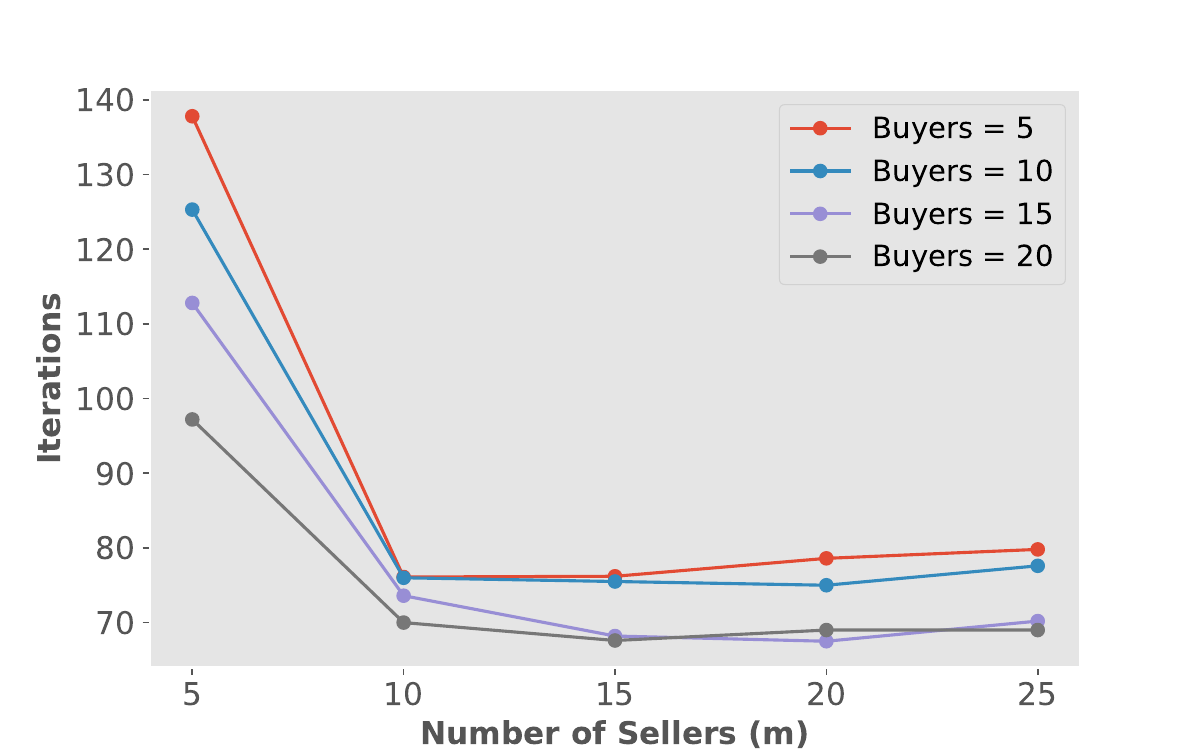}
        \caption{Convergence of Inertial Approximate Best Response}
        \label{fig:iterations}
    \end{subfigure}
    \begin{subfigure}[t]{0.43\textwidth}
        \centering
        \includegraphics[width=\linewidth]{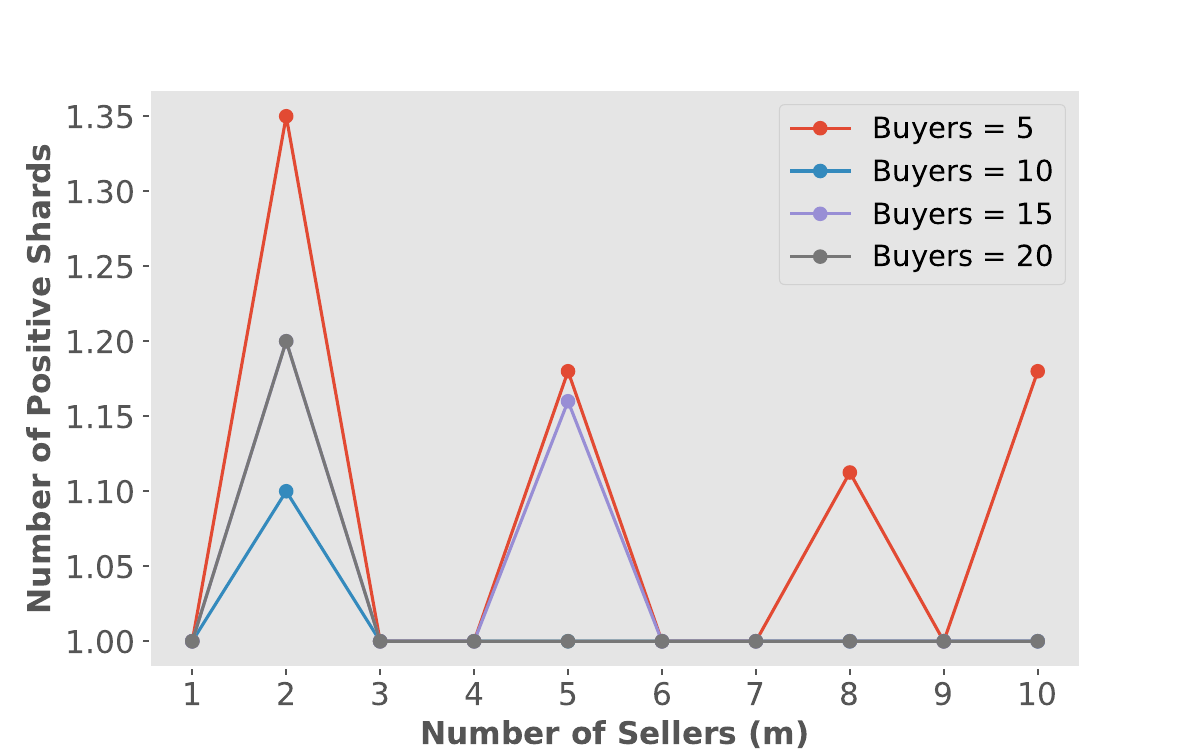}
        \caption{Average Sparsity}
        \label{fig:shards}
    \end{subfigure}
\end{figure*}

\paragraph{Dynamics: Inertial Approximate Best Response.}
Given the proof of approximate Nash equilibrium in \Cref{sec:approx-ne}, we study a natural \emph{inertial response}: each seller $j$ in round $t$, finds the PLC response that guarantees her revenue at least half of her randomized best linear response as defined in Section~\ref{sec:approx-ne}. Recall that this is obtained by solving the following program
\begin{align*}
\max_{\bm{x}}  \sum_{k=1}^{m_j} x_k\cdot r_j(\vecE^{(k)}, \ell_{-j}) \quad \text{subject to} \quad \bm{x} \in \Delta_{m_j} \,.
\end{align*}
and defining a PLC function $\bm \ell^*_j$, by setting $\ell^*_{j,k} = x^*_k$, where $\bm x^*$ is a optimal solution to the above program. Then, the seller updates her response in round $t$ as $\ell_j^t$ as  $\ell_j^{t-1} \cdot (1-\alpha)+ \bm{\ell}^*_j\cdot \alpha$, where $\alpha$ is the learning rate controlling the inertia. In our implementation, we set $\alpha = 0.1$, all sellers start with a strategy profile $\ell^0_{j, k} = 1/m_j$ for every $j, k$, and terminate the dynamics when $\norm{\ell^{t+1} - \ell^t}_{\infty}$ is no more than $\eps = 10^{-4}$.

\paragraph{Convergence Results and Insights.}
We evaluated the inertial approximate best-response dynamics on the news data marketplace,
varying the number of buyers $n \in \{5,10,15,20\}$ and the number of sellers $m \in \{5,10,15,20,25\}$.
Each buyer has a unit budget, and each seller employs a PLC pricing function
with the set $\{k/10 \mid k \in [10]\}$ of discretized prices.
As shown in \Cref{fig:iterations}, across all instances the dynamics converge efficiently,
terminating within at most $137$ rounds.

Although convergence is consistently fast, the speed of convergence is not strictly monotone in the number of sellers. In particular, while the instance with $m=5$ is smaller in size than the instance with $m=10$, it requires a relatively larger number of iterations to converge. This phenomenon can be attributed to \emph{stronger strategic interdependence in smaller markets}: when the number of sellers is small, a single seller’s update induces a larger change in the revenue obtained by other sellers, causing their subsequent responses to vary more significantly and leading to slower stabilization of the dynamics.

As the number of sellers increases, the marginal impact of any individual seller’s update diminishes, reflecting a \emph{large-market} effect. Once this effect becomes prominent, the dynamics stabilize faster and the convergence rate improves sharply. Beyond this point, we observe a steady and approximately monotone increase in the speed of convergence as the number of sellers continues to grow.

Overall, the empirical results exhibit a non-monotone convergence pattern: convergence is initially slower in very small markets due to strong strategic coupling among sellers, followed by a marked improvement as large-market effects emerge, and finally a smooth and steady growth as the market size increases.

\paragraph{Approximation Guarantees and Insights.} To quantify the approximation quality of a strategy profile $\bm \ell$, we define its \emph{incentive ratio} as
\begin{align*}
\mathrm{incentive\ ratio}(\bm \ell)
= \max_{j\in [m]}
\frac{\max_{k} r_j(\mathbf{e}^{(k)},  \ell_{-j})}
     {r_j( \ell_j,  \ell_{-j})},
\end{align*}
which measures the maximum multiplicative gain that any seller can obtain by unilaterally deviating to any linear pricing strategy.

Across all experimental settings, spanning a wide range of buyer and seller counts, we consistently observe that the incentive ratio equals $1$. That is, despite our theoretical guarantees establishing only the existence of $2$-Nash equilibria, the proposed dynamics reliably converge to \emph{exact} Nash equilibria with respect to linear pricing strategies. Moreover, this phenomenon persists across diverse market sizes, indicating that it is not a fragile artifact of small instances or particular parameter choices.

These findings highlight a gap between worst-case theoretical guarantees and typical empirical behavior, and naturally motivate further theoretical investigation. In particular, they raise the possibility that pricing games arising in real-world data marketplaces exhibit additional structural properties that are not captured by the worst-case analysis, and which may be leveraged to strengthen equilibrium guarantees.

\paragraph{Sparsity of the PLC Strategies.} A PLC pricing strategy is called \emph{sparse} if it consists of only a few shards with positive length. Sparse strategies are easier to interpret and implement. We measure sparsity as the number of shards with strictly positive length and report the average sparsity of equilibrium pricing functions across market sizes in \Cref{fig:shards}.

Across all instances, the average sparsity is at most two and always well below the number of buyers or sellers, indicating that equilibrium PLC strategies are nearly linear. Moreover, sparsity increases slightly when there are fewer sellers, reflecting stronger strategic pricing incentives, while it decreases as the number of sellers grows. In particular, when the number of sellers exceeds two, sellers rarely employ more than a single shard. This decline in sparsity can be attributed similarly to the \emph{large-market effect}: as competition intensifies, the influence of any single seller on market outcomes diminishes, reducing the incentive to employ complex, multi-shard pricing strategies. Motivated by our observations on sparsity, we further study the sparsity of PLC best responses and show in Appendix~\ref{sec:no-of-shards-extra} that for any seller, there always exists a PLC best response with sparsity at most the number of buyers.

\ifx\compatmode\undefined

\else
\fi

\addcontentsline{toc}{section}{References}

\appendix
\crefalias{section}{appendix}
\crefalias{subsection}{appendix}
\crefalias{subsubsection}{appendix}

\section{Details on the Non-Existence of (Approximate) Linear NE}
\label{sec:lne-cex-extra}

In \cref{sec:lne-cex}, we claimed that that if sellers only use linear pricing,
then no $1.363$-approximate pure Nash equilibrium exists.
In this appendix, we provide a formal proof of this claim.

Recall the data market instance we constructed in \cref{ex:lne-cex}, parametrized by $(\alpha, \beta)$.

\subsection{Characterizing Revenues}
\label{sec:lne-cex:rev}

Let $r_1(p, q)$ and $r_2(p, q)$ be the revenue earned by sellers 1 and 2, respectively,
when seller 1 prices her dataset at $p$ and seller 2 prices her dataset at $q$.
Let $r_1^*(q)$ be the maximum revenue seller 1 can make when seller 2 prices her dataset at $q$.
Define $r_2^*(p)$ analogously.
In \cref{thm:lne-cex:r,thm:lne-cex:rstar}, we give closed-form expressions
for $r_1(p, q)$, $r_2(p, q)$, $r_1^*(q)$, and $r_2^*(p)$.

\begin{lemma}
\label{thm:lne-cex:r}
We have
\begin{align*}
(n-1)\min(p, \max(0, 1-q))
    &\le r_1(p, q) - \min(\beta, p) \le (n-1)\min(p, 1),
\\ (n-1)\min(q, \max(0, 1-p))
    &\le r_2(p, q) \le (n-1)\min(q, 1).
\end{align*}
Moreover,
\begin{enumerate}
\item If $p > \alpha q$, then $r_1(p, q) = \min(\beta, p) + (n-1)\min(p, \max(0, 1-q))$
    and $r_2(p, q) = (n-1)\min(q, 1)$.
\item If $p < \alpha q$, then $r_1(p, q) = \min(\beta, p) + (n-1)\min(p, 1)$
    and $r_2(p, q) = (n-1)\min(q, \max(0, 1-p))$.
\item If $p + q \le 1$, then $r_1(p, q) = \min(\beta, p) + (n-1)p$ and $r_2(p, q) = (n-1)q$.
\item If $p + q \ge 1$, then $r_1(p, q) + r_2(p, q) = (n-1) + \min(\beta, p)$.
\end{enumerate}
\end{lemma}
\begin{proof}
Suppose sellers 1 and 2 price their datasets at $p$ and $q$, respectively.
The rich buyer has a positive value only for dataset 1, so she will purchase as much of it as she can.
Thus, seller 1 gets a revenue of $\min(\beta, p)$ from the rich buyer
and seller 2 doesn't earn anything from her.

If $p < \alpha q$, then dataset 1 gives a poor buyer her maximum bang-per-buck, but dataset 2 doesn't.
So a poor buyer will purchase as much of dataset 1 as possible before moving on to dataset 2.
Thus, seller 1's revenue from poor buyers would be $(n-1)\min(p, 1)$,
and seller 2's revenue from poor buyers would be $(n-1)\min(q, 1-\min(p, 1)) = (n-1)\min(q, \max(0, 1-p))$.

If $p > \alpha q$, then dataset 2 gives a poor buyer her maximum bang-per-buck, but dataset 1 doesn't.
So a poor buyer will purchase as much of dataset 2 as possible before moving on to dataset 1.
Thus, seller 2's revenue from poor buyers would be $(n-1)\min(q, 1)$,
and seller 1's revenue from poor buyers would be $(n-1)\min(p, 1-\min(q, 1)) = (n-1)\min(p, \max(0, 1-q))$.

If $p + q \le 1$, then a poor buyer will purchase both datasets completely
(regardless of which gives her more bang-per-buck).
Thus, sellers 1 and 2 will earn revenues of $(n-1)p$ and $(n-1)q$, respectively, from poor buyers.

If $p + q \ge 1$, then a poor buyer would exhaust her budget if she purchases both datasets completely.
Thus, the sum of the revenues earned by sellers 1 and 2 is $n-1$.

If $p + q > 1$ and $p = \alpha q$, then both datasets give the same bang-per-buck to a poor buyer.
The sellers' revenues depend on how the buyer splits her budget among them.
For seller 1, the worst case is if a poor buyer purchases as much of dataset 2
as she can before moving on to dataset 1, in which case
seller 1's revenue from poor buyers would be $(n-1)\min(p, \max(0, 1-q))$.
Seller 1's best case is if a poor buyer purchases as much of dataset 1
as she can before moving on to dataset 2,in which case
seller 1's revenue from poor buyers would be $(n-1)\min(p, 1)$.
Thus, $(n-1)\min(p, \max(0, 1-q)) \le r_1(p, q) - \min(\beta, p) \le (n-1)\min(p, 1)$.
Similarly, $(n-1)\min(q, \max(0, 1-p)) \le r_2(p, q) \le (n-1)\min(q, 1)$.
\end{proof}

\begin{remark}
\label{remark:lne-cex:pq-ambig}
Suppose sellers 1 and 2 price their datasets at $p$ and $q$, respectively.
If $p + q > 1$ and $p = \alpha q$, then both datasets are MBB for buyer 1,
so the sellers' revenues cannot be precisely determined,
since they depend on how buyer 1 splits her budget among them.
Thus, $r_1$ and $r_2$ need more inputs in addition to just $p$ and $q$ to be well-defined.
Thus, treating $r_1$ and $r_2$ like functions of $p$ and $q$ is technically incorrect.
However, the bulk of our analysis does not run into the case $p = \alpha q$ and $p + q > 1$,
and we will handle this case separately in the end.
\end{remark}

\begin{lemma}
\label{thm:lne-cex:rstar}
\begin{align*}
r_1^*(q) &= \max\left(
    \beta + (n-1)\max(0, 1-q),
    \min(\beta, \alpha q) + (n-1)\min(1, \alpha q)
    \right),
\\ r_2^*(p) &= (n-1)\max(1-p, \min(1, p/\alpha)).
\end{align*}
\end{lemma}
\begin{proof}
We have
\begin{align*}
r_1^*(q) &\defeq \sup_{p \in \mathbb{R}_{\ge 0}} r_1(p, q),
& r_2^*(q) &\defeq \sup_{q \in \mathbb{R}_{\ge 0}} r_2(p, q).
\end{align*}

By \cref{thm:lne-cex:r}, if $q < p/\alpha$, then
$r_2(p, q) = (n-1)\min(q, 1) \le (n-1)\min(p/\alpha, 1)$,
and $r_2(p, p/\alpha) \le (n-1)\min(p/\alpha, 1)$.
Moreover, for any $0 < \eps \le p/\alpha$, we have
$r_2(p, p/\alpha-\eps) = (n-1)\min(p/\alpha-\eps, 1)
\ge (n-1)\min(p/\alpha, 1) - (n-1)\eps$. Thus,
\[ \sup_{q \in [0, p/\alpha]} r_2(p, q) = (n-1)\min(p/\alpha, 1). \]
If $q > p/\alpha$, then $r_2(p, q) = (n-1)\min(q, \max(0, 1-p))$,
which is non-decreasing in $q$. Thus,
\[ \sup_{q \in (p/\alpha, \infty)} r_2(p, q) = r_2(p, \infty) = (n-1)\max(0, 1-p). \]
Hence,
\begin{align*}
r_2^*(p) &\defeq \sup_{q \in \mathbb{R}_{\ge 0}} r_2(p, q)
\\ &= \max\left(\sup_{q \in [0, p/\alpha]} r_2(p, q), \sup_{q \in (p/\alpha, \infty)} r_2(p, q)\right)
\\ &= (n-1)\max(\min(p/\alpha, 1), \max(0, 1-p))
\\ &= (n-1)\max(1-p, \min(1, p/\alpha)).
\end{align*}

By \cref{thm:lne-cex:r}, if $p < \alpha q$, then
$r_1(p, q) \le \min(\beta, p) + (n-1)\min(p, 1) \le \min(\beta, \alpha q) + (n-1)\min(\alpha q, 1)$,
and $r_1(\alpha q, q) \le \min(\beta, \alpha q) + (n-1)\min(\alpha q, 1)$.
Moreover, for any $0 < \eps \le \alpha q$, we have
$r_1(\alpha q - \eps, q) = \min(\beta, \alpha q - \eps) + (n-1)\min(\alpha q - \eps, 1)
    \ge \min(\beta, \alpha q) + (n-1)\min(\alpha q, 1) - n\eps$.
Thus,
\[ \sup_{p \in [0, \alpha q]} r_1(p, q) = \min(\beta, \alpha q) + (n-1)\min(1, \alpha q). \]
For $p > \alpha q$, we have $r_1(p, q) = \min(\beta, p) + (n-1)\min(p, \max(0, 1-q))$,
which is non-decreasing in $p$. Thus,
\[ \sup_{p \in (\alpha q, \infty)} r_1(p, q) = r_1(\infty, q) = \beta + (n-1)\max(0, 1-q). \]
Hence,
\begin{align*}
r_1^*(q) &\defeq \sup_{p \in \mathbb{R}_{\ge 0}} r_1(p, q)
\\ &= \max\left(\sup_{p \in [0, \alpha q]} r_1(p, q), \sup_{q \in (\alpha q, \infty)} r_1(p, q)\right)
\\ &= \max\left(
    \min(\beta, \alpha q) + (n-1)\min(1, \alpha q),
    \beta + (n-1)\max(0, 1-q)
    \right).
\qedhere
\end{align*}
\end{proof}

\begin{remark}
$r_1^*(q)$ and $r_2^*(p)$ \emph{only} depend on $q$ and $p$, respectively
(unlike $r_1$ and $r_2$; see \cref{remark:lne-cex:pq-ambig}).
Intuitively, this is because undercutting the other seller by an infinitesimally small amount
ensures that the revenue is not ambiguous.
\end{remark}

\subsection{Constraints and Special Points}
\label{sec:lne-cex:constr-and-points}

The inapproximability factor is given by
\[ \inf_{p, q \in \mathbb{R}_{\ge 0}} \mu(p, q), \]
where
\[ \mu(p, q) \defeq \max\left(\frac{r_1^*(q)}{r_1(p, q)}, \frac{r_2^*(p)}{r_2(p, q)}\right). \]
(Note that $r_1^*(q) \ge \beta > 0$ for all $q \ge 0$, and $r_2^*(p) > 0$ for all $p \ge 0$,
and so, if $r_1(p, q) = 0$ or $r_2(p, q) = 0$, we define $\mu(p, q) = \infty$.)

To make the analysis simpler, we impose some constraints on $(\alpha, \beta)$
instead of searching the entire region $\mathbb{R}_{\ge 0}^2$.
Under these constraints, we identify a set of four points $P \defeq \{(p_i, q_i)\}_{i=1}^4$
such that $\argmin_{p,q} \mu(p, q) \in P$.
Then, we pick $(\alpha, \beta)$ such that $\min_{i=1}^4 \mu(p_i, q_i)$ is maximized.

We first simply state the constraints and the points $\{(p_i, q_i)\}_{i=1}^4$.
After that, we explain how we came up with them and prove that $\argmin_{p,q} \mu(p, q) \in P$.

\textbf{Constraints}:
\begin{enumerate}
\item \label{item:lne-cex:c1}$2 \le \alpha \le \beta < \alpha + n-1$.
\item \label{item:lne-cex:c2}$\alpha + n-1 < \alpha\beta$.
\item \label{item:lne-cex:c3}$(\alpha + n-1)^3 > \beta(\alpha\beta + 2(n-1)(\alpha + n-1))$.
\item \label{item:lne-cex:c4}$\alpha(\beta + n-1)^2 > \beta(n-1)^2$.
\end{enumerate}

Let $L_1 \defeq \alpha + (n-1) - \beta$ and $L_2 \defeq (n-1)^2 + \alpha(n-1) + \alpha\beta$.
Then $L_1 > 0$ by Constraint \ref{item:lne-cex:c1}.

\textbf{Points}:
\begin{align*}
q_1 &\defeq \frac{\beta}{\alpha + n-1},
    & \chat_1 &\defeq \frac{n-1 + \alpha q_1}{n},
    & p_1 &\defeq \frac{2}{1+\sqrt{1+4/\alpha\chat_1}},
    & \mu_1 &\defeq \frac{\chat_1}{p_1}.
\end{align*}
\begin{align*}
p_2 &\defeq \beta
    & q_2 &\defeq \frac{n-1+\beta}{n-1 + \sqrt{L_2}},
    & \mu_2 &\defeq \frac{1}{q_2}.
\end{align*}
\begin{align*}
\mu_3 &\defeq \frac{\sqrt{(n-1)^2L_1^2 + 4\alpha\beta L_2} - L_1(n-1)}{2\alpha\beta},
    & p_3 &\defeq \alpha q_1 \mu_3,
    & q_3 &\defeq q_1.
\end{align*}
\begin{align*}
q_4 &\defeq q_1,
    & p_4 &\defeq \alpha q_1,
    & \mu_4 &\defeq 1 + \frac{\beta(n-1)}{L_2}.
\end{align*}

Our main claim is that it suffices to only care about the points $\{(p_i, q_i)\}_{i=1}^4$.

\begin{lemma}
\label{thm:lne-cex:pq-redn}
If $(\alpha, \beta)$ satisfies Constraints \ref{item:lne-cex:c1} to \ref{item:lne-cex:c4}, then
\[ \inf_{p,q \in \mathbb{R}_{\ge 0}} \mu(p, q) = \min_{i=1}^4 \mu(p_i, q_i). \]
Moreover, $\mu(p_i, q_i) = \mu_i$ for all $i \in [4]$.
\end{lemma}

To prove as strong an inapproximability as possible,
we must pick $(\alpha, \beta)$ such that $c^* \defeq \min(\mu_1, \mu_2, \mu_3, \mu_4)$ is maximized.
For $\alpha = 0.7331568565(n-1)$, $\beta = 0.8603994358(n-1)$, and $n \to \infty$,
we get $c^* \approx 1.3639646021$. By plugging these values into \cref{thm:lne-cex:pq-redn},
we prove \cref{thm:lne-cex}.

In \cref{sec:lne-cex:sp-props}, we prove useful properties of the points $\{(p_i, q_i)\}_{i=1}^4$
and show that $\mu(p_i, q_i) = \mu_i$ for all $i \in [4]$.
In \cref{sec:lne-cex:case1,sec:lne-cex:case2,sec:lne-cex:case3,sec:lne-cex:case4},
we prove \cref{thm:lne-cex:pq-redn} (split across
\cref{thm:lne-cex:1.1,thm:lne-cex:1.2,thm:lne-cex:2,thm:lne-cex:3,thm:lne-cex:4.1,thm:lne-cex:4.2}).

\subsection{Properties of the Special Points}
\label{sec:lne-cex:sp-props}

We now hint towards where we got the points $\{(p_i, q_i)\}_{i=1}^3$ from,
and prove that $\mu_i = \mu(p_i, q_i)$ for all $i \in [3]$.
We defer the corresponding discussion for $(p_4, q_4)$ to \cref{sec:lne-cex:case4},
since $p_4/q_4 = \alpha$, and that case needs careful handling.

\begin{lemma}
\label{thm:lne-cex:q1}
$q_1 \in (1/\alpha, 1)$, $r_1^*(q_1) = n\chat_1$,
and $r_1^*(q) \ge n\chat_1$ for all $q \ge 0$.
\end{lemma}
\begin{proof}
By Constraint \ref{item:lne-cex:c1}, $q_1 = \beta/(\alpha + n-1) < 1$.
By Constraint \ref{item:lne-cex:c2}, $\alpha q_1 = \alpha\beta/(\alpha + n-1) > 1$.
Thus, $q_1 \in (1/\alpha, 1)$.

Define the functions $g_1, g_2: \mathbb{R}_{\ge 0} \to \mathbb{R}$ as follows:
\begin{align*}
g_1(x) &\defeq \beta + (n-1)\max(0, 1-x),
& g_2(x) &\defeq \min(\beta, \alpha x) + (n-1)\min(1, \alpha x).
\end{align*}
Then
\begin{align*}
g_1(q_1) &= \beta + (n-1)(1-q_1) = \beta + (n-1)\left(1 - \frac{\beta}{\alpha + n-1}\right)
\\ &= (n-1) + \beta\left(1 - \frac{n-1}{\alpha + n-1}\right)
    = n-1 + \frac{\alpha\beta}{\alpha + n-1}
\\ &= n-1 + \alpha q_1 = n\chat_1.
\end{align*}
By Constraint \ref{item:lne-cex:c1}, $\alpha q_1 < \alpha \le \beta$.
Since $q_1 > 1/\alpha$, we get
$g_2(q_1) = \min(\beta, \alpha q_1) + (n-1)\min(1, \alpha q_1) = \alpha q_1 + n-1 = n\chat_1$.
Hence, $g_1(q_1) = g_2(q_1)$.

$r_1^*(q_1) = \max(g_1(q_1), g_2(q_1)) = n\chat_1$.
$g_1$ is non-increasing, and $g_2$ is non-decreasing.
Thus, for any $q \in [0, q_1]$, we have
$r_1^*(q) = \max(g_1(q), g_2(q)) \ge g_1(q) \ge g_1(q_1) = n\chat_1$,
and for any $q \ge q_1$, we have
$r_1^*(q) = \max(g_1(q), g_2(q)) \ge g_2(q) \ge g_2(q_1) = n\chat_1$.
Thus, $r_1^*(q) \ge n\chat_1$ for all $q \ge 0$.
\end{proof}

\begin{lemma}
\label{thm:lne-cex:p1}
$p_1 \in (\frac{\alpha}{\alpha+1}, 1)$ and $\chat_1 > 1$.
\[ \frac{r_1^*(q_1)}{r_1(p_1, q_1)} = \frac{r_2^*(p_1)}{r_2(p_1, q_1)} = \mu_1, \]
so $\mu(p_1, q_1) = \mu_1$.
Moreover, $p_1$ is the unique positive solution to the equation
\[ \frac{\chat_1}{x} = \frac{x}{\alpha(1-x)}. \]
\end{lemma}
\begin{proof}
By \cref{thm:lne-cex:q1}, we have $q_1 \in (1/\alpha, 1)$.
$\chat_1 = (\alpha q_1 + (n-1))/n > (1 + (n-1))/n = 1$.

Define the function $f: (0, 1) \to \mathbb{R}$ as
\[ f(x) \defeq \frac{\chat_1}{x} - \frac{x}{\alpha(1-x)}. \]
It is easy to verify that $p_1$ is the unique positive solution to $f(x) = 0$.

$f$ is a (strictly) decreasing function,
\[ f\left(\frac{\alpha}{\alpha+1}\right) = \chat_1\frac{\alpha+1}{\alpha} - 1 > \frac{1}{\alpha} > 0, \]
and $\lim_{x \to 1} f(x) = -\infty$.
Thus, $f$ has a unique zero in the interval $(\frac{\alpha}{\alpha+1}, 1)$.
Thus, $p_1 \in (\frac{\alpha}{\alpha+1}, 1)$.

Since $p_1 \in (\frac{\alpha}{\alpha+1}, 1)$ and $q_1 \in (1/\alpha, 1)$,
we get $p_1 < \alpha q_1$, $p_1 + q_1 \ge 1$, and $1-p_1 \le p_1/\alpha$.
Moreover, by \cref{thm:lne-cex:q1}, we have $r_1^*(q_1) = n\chat_1$. Thus,
\begin{align*}
\frac{r_1^*(q_1)}{r_1(p_1, q_1)} &= \frac{\chat_1}{p_1} = \mu_1,
& \frac{r_2^*(p_1)}{r_2(p_1, q_1)} &= \frac{p_1/\alpha}{1-p_1} = \frac{\chat_1}{p_1} = \mu_1.
\qedhere \end{align*}
\end{proof}

\begin{lemma}
\label{thm:lne-cex:q2}
$q_2 \in (1/\alpha, 1)$, $\mu_2 = \mu(\beta, q_2)$, and
\[ \frac{r_1^*(q_2)}{r_1(\beta, q_2)} = \frac{r_2^*(\beta)}{r_2(\beta, q_2)} = \mu_2. \]
Moreover, $q_2$ is the unique positive solution to the equation
\[ \frac{\alpha x + n-1}{\beta + (n-1)(1-x)} = \frac{1}{x}. \]
\end{lemma}
\begin{proof}
Define the function $f: (0, 1] \to \mathbb{R}$ as
\[ f(x) \defeq \frac{\alpha x + n-1}{\beta + (n-1)(1-x)} - \frac{1}{x}. \]
It is easy to verify that $q_2$ is the unique positive solution to $f(x) = 0$.

By Constraint \ref{item:lne-cex:c1},
\[ f(1) = \frac{\alpha + n-1}{\beta} - 1 > 0. \]
\begin{align*}
f(1/\alpha) &= \frac{n}{\beta + (n-1)(1-1/\alpha)} - \alpha
\\ &= \frac{\alpha n}{\alpha\beta + (n-1)(\alpha - 1)} - \alpha
\\ &\le \frac{\alpha n}{4 + (n-1)} - \alpha
    \tag{since $\alpha\beta \ge 4$ and $\alpha-1 \ge 1$}
\\ &= -\frac{3\alpha}{n+3} < 0.
\end{align*}
Since $f$ is a (strictly) increasing function, $f(1) > 0$, and $f(1/\alpha) < 0$,
we get that it has a unique zero in $(1/\alpha, 1)$.
Thus, $q_2 \in (1/\alpha, 1)$.

We have $p_2 \defeq \beta \ge \alpha > \alpha q_2$, and $p_2 + q_2 \ge \beta > 1$. Thus,
\begin{align*}
\frac{r_1^*(q_2)}{r_1(\beta, q_2)} &= \frac{\max(\beta + (n-1)(1-q_2), \alpha q_2 + n-1)}{\beta + (n-1)q_2}
\\ &= \max\left(1, \frac{\alpha q_2 + n-1}{\beta + (n-1)(1-q_2)}\right)
\\ &= \max\left(1, \frac{1}{q_2}\right)
    \tag{since $f(q_2) = 0$}
\\ &= \mu_2,
\\ \frac{r_2^*(\beta)}{r_2(\beta, q_2)} &= \frac{1}{q_2} = \mu_2.
\end{align*}
Hence, $\mu(\beta, q_2) = \mu_2$.
\end{proof}

\begin{lemma}
\label{thm:lne-cex:mu3}
$1 < \mu_3 < 1/q_1$ and
$\mu_3$ is the unique positive solution to the equation
$\alpha\beta x^2 + (n-1)(n-1+\alpha-\beta)x - ((n-1)^2 + (n-1)\alpha + \alpha\beta) = 0$.
\end{lemma}
\begin{proof}
Let $f(x) = \alpha\beta x^2 - (n-1)(n-1+\alpha-\beta)x - ((n-1)^2 + (n-1)\alpha + \alpha\beta)$.
It is easy to verify that $\mu_3$ is the unique positive solution to $f(x) = 0$.

Let $\nu$ be the other (negative) root of $f(x) = 0$.
Since $f$ is a quadratic function, where the quadratic term's coefficient is positive,
we get that $f(x) \le 0 \iff x \in [\nu, \mu_3]$.
Thus, to prove that $1 < \mu_3 < 1/q_1$,
we must show that $f(1) < 0$ and $f(1/q_1) > 0$.
$f(1) = -\beta(n-1) < 0$.
One can verify that $f(1/q_1) > 0$ is equivalent to Constraint \ref{item:lne-cex:c3}.

Thus, $f(1/q_1) > 0$, and so, $1 < \mu_3 < 1/q_1$.
\end{proof}

\begin{lemma}
\label{thm:lne-cex:p3}
$p_3 \in (\alpha q_3, \alpha)$, $\mu_3 = \mu(p_3, q_3)$, and
\[ \frac{r_1^*(q_3)}{r_1(p_3, q_3)} = \frac{r_2^*(p_3)}{r_2(p_3, q_3)} = \mu_3. \]
\end{lemma}
\begin{proof}
$q_3 \defeq q_1$ and $p_3 \defeq \alpha q_1 \mu_3$.
By \cref{thm:lne-cex:mu3}, $\mu_3 \in (1, 1/q_1)$, so $p_3 \in (\alpha q_1, \alpha)$.

By \cref{thm:lne-cex:q1}, $q_1 \in (1/\alpha, 1)$, so $p_3 \ge \alpha q_1 \ge 1$.
Since $p_3 \in (\alpha q_3, \alpha)$, we get
\[ \frac{r_2^*(p_3)}{r_2(p_3, q_3)} = \frac{p_3/\alpha}{q_1} = \mu_3. \]
By \cref{thm:lne-cex:q1}, $r_1^*(q_1) = n\chat_1 = \alpha q_1 + n-1$. Thus,
\begin{align*}
\frac{r_1^*(q_3)}{r_1(p_3, q_3)} &= \frac{\alpha q_1 + n-1}{p_3 + (n-1)(1-q_1)}
\\ &= \frac{\alpha\frac{\beta}{\alpha+n-1} + n-1}{\alpha \mu_3 \frac{\beta}{\alpha+n-1}
    + (n-1)\left(1 - \frac{\beta}{\alpha + n-1}\right)}
\\ &= \frac{(n-1)^2 + (n-1)\alpha + \alpha\beta}{\alpha\beta\mu_3 + (n-1)(\alpha + n-1 - \beta)}.
\end{align*}
Since $\mu_3$ is a positive root of the equation
$\alpha\beta x^2 + (n-1)(\alpha + n-1 - \beta)x - ((n-1)^2 + (n-1)\alpha + \alpha\beta) = 0$,
we get that $r_1^*(q_3)/r_1(p_3, q_3) = \mu_3$.
Thus, $\mu(p_3, q_3) = \mu_3$.
\end{proof}

\subsection{\texorpdfstring{\boldmath}{}Case 1:
    \texorpdfstring{$p > \alpha q$}{p > aq}
    and \texorpdfstring{$p + q \ge 1$}{p + q >= 1}}
\label{sec:lne-cex:case1}

Recall \cref{thm:lne-cex:r,thm:lne-cex:rstar}.
Since $p > \alpha q$ and $p + q \ge 1$, we get
\begin{align*}
r_1^*(q) &= \max\left(\beta + (n-1)\max(0, 1-q),
    \min(\beta, \alpha q) + (n-1)\min(1, \alpha q)\right),
    \tag{by \cref{thm:lne-cex:rstar}}
\\ r_1(p, q) &= \min(\beta, p) + (n-1)\max(0, 1-q),
    \tag{by \cref{thm:lne-cex:r}, $p + q \ge 1$}
\\ r_2^*(p) &= (n-1)\min(1, p/\alpha),
    \tag{by \cref{thm:lne-cex:rstar}, $1-p \le q < p/\alpha$}
\\ r_2(p, q) &= (n-1)\min(q, 1).
    \tag{by \cref{thm:lne-cex:r}}
\end{align*}

\begin{lemma}
\label{thm:lne-cex:1.1}
Let $p + q \ge 1$ and $p > \alpha q$.
If $p \ge \alpha$, then $\mu(p, q) \ge \mu_2$.
\end{lemma}
\begin{proof}
By \cref{thm:lne-cex:q2}, $q_2 \in (1/\alpha, 1)$.

\textbf{Case 1}: $q \le q_2$. Since $p \ge \alpha$, we get
\[ \mu(p, q) \ge \frac{r_2^*(p)}{r_2(p, q)} = \frac{1}{q} \ge \frac{1}{q_2} = \mu_2. \]

\textbf{Case 2}: $q \ge q_2$. Then $\alpha q \ge \alpha q_2 > 1$. So,
\begin{align*}
\mu(p, q) &\ge \frac{r_1^*(q)}{r_1(p, q)}
    \ge \frac{\min(\beta, \alpha q) + n-1}{\min(\beta, p) + (n-1)\max(0, 1-q)}
\\ &\ge \frac{\alpha q_2 + n-1}{\beta + (n-1)(1-q_2)}
        \tag{expression was non-decreasing in $q$, and $\beta \ge \alpha > \alpha q_2$}
\\ &= \frac{1}{q_2} = \mu_2.
        \tag{by \cref{thm:lne-cex:q2}}
\end{align*}

Thus, $\mu(p, q) \ge \mu_2$.
\end{proof}

\begin{lemma}
\label{thm:lne-cex:mu3-2}
$\mu_3 < 1 + \beta/(n-1)$.
\end{lemma}
\begin{proof}
By \cref{thm:lne-cex:mu3}, $\mu_3$ is the only positive solution to the equation $f(x) = 0$,
where $f(x) \defeq \alpha\beta x^2 + (n-1)(n-1+\alpha-\beta)x - ((n-1)^2 + (n-1)\alpha + \alpha\beta)$.

Let $\nu$ be the other (negative) root of $f(x) = 0$.
Since $f$ is a quadratic function, where the quadratic term's coefficient is positive,
we get that $f(x) \le 0 \iff x \in [\nu, \mu_3]$.
Thus, to prove that $\mu_3 < 1 + \beta/(n-1)$,
we must show that $f(1 + \beta/(n-1)) > 0$.
One can verify that $f(1 + \beta/(n-1)) > 0$ is equivalent to Constraint \ref{item:lne-cex:c4}.
Thus, $\mu_3 < 1 + \beta/(n-1)$.
\end{proof}

\begin{lemma}
\label{thm:lne-cex:1.2}
Let $p + q \ge 1$ and $p > \alpha q$.
If $p \le \alpha$, then $\mu(p, q) \ge \mu_3$.
\end{lemma}
\begin{proof}
Since $\alpha q < p \le \alpha$, we get $q < 1$. Since $p + q \ge 1$, we have $p > 0$.
If $q = 0$, then $r_2(p, q) = 0$, and then $\mu(p, q) = \infty$. So let $q > 0$.
Since $p \le \alpha$, we get $r_2^*(p)/r_2(p, q) = p/(\alpha q)$.
If $p \ge \alpha q \mu_3$, then $\mu(p, q) \ge p/(\alpha q) \ge \mu_3$, and we are done.
So now assume $p < \alpha \mu_3 q$.

Recall that $q_1 = \beta/(\alpha+n-1)$. Assume $\mu(p, q) < \mu_3$.

\textbf{Case 1}: $q \le q_1$. Then
\begin{align*}
& \mu_3 > \mu(p, q) \ge \frac{r_1^*(q)}{r_1(p, q)} \ge \frac{\beta + (n-1)(1-q)}{\alpha q \mu_3 + (n-1)(1-q)}
\\ &\implies \mu_3(\alpha q \mu_3 + (n-1)(1-q)) > \beta + (n-1)(1-q)
\\ &\implies (\alpha\mu_3^2 - (n-1)\mu_3 + (n-1))q > \beta + n-1 - (n-1)\mu_3
\\ &\implies (\alpha\mu_3^2 - (n-1)\mu_3 + (n-1))(q - q_1)
    \\ &\qquad\qquad> (\beta + n-1) - (n-1)\mu_3 - (\alpha\mu_3^2 - (n-1)\mu_3 + (n-1))q_1
    \\ &\qquad\qquad= -(\alpha\beta\mu_3^2 + (n-1)(\alpha+n-1-\beta)\mu_3 - (\alpha\beta + (n-1)\alpha + (n-1)^2))
    \\ &\qquad\qquad= 0.
        \tag{by \cref{thm:lne-cex:mu3}}
\end{align*}
Thus, we have $(\alpha\mu_3^2 - (n-1)\mu_3 + (n-1))(q - q_1) > 0$.
Moreover, since $\mu_3 > 1 + \beta/(n-1)$ by \cref{thm:lne-cex:mu3-2}, we get
$\alpha\mu_3^2 - (n-1)\mu_3 + (n-1) = ((\beta + n-1) - (n-1)\mu_3)/q_1 > 0$.
Hence, $q > q_1$, which is a contradiction.

\textbf{Case 2}: $q \ge q_1$.
\\ Since $q_1 \in (1/\alpha, 1)$ and $q < 1$, we get $1 < \alpha q < \alpha \le \beta$.
Thus,
\begin{align*}
& \mu_3 > \mu(p, q) \ge \frac{r_1^*(q)}{r_1(p, q)} \ge \frac{\alpha q + n-1}{\alpha q \mu_3 + (n-1)(1-q)}
\\ &\implies \mu_3(\alpha q \mu_3 + (n-1)(1-q)) > \alpha q + n-1
\\ &\implies (n-1)(\mu_3 - 1) > (\alpha + (n-1)\mu_3 - \alpha\mu_3^2)q
\\ &\implies (\alpha + (n-1)\mu_3 - \alpha\mu_3^2)(q_1-q)
    \\ &\qquad\qquad> (\alpha + (n-1)\mu_3 - \alpha\mu_3^2)q_1 - (n-1)(\mu_3 - 1)
    \\ &\qquad\qquad= (\alpha + (n-1)\mu_3 - \alpha\mu_3^2)q_1 - (n-1)(\mu_3 - 1)
    \\ &\qquad\qquad= -(\alpha\beta\mu_3^2 + (n-1)(\alpha+n-1-\beta)\mu_3 - ((n-1)^2 + \alpha(n-1) + \alpha\beta))
    \\ &\qquad\qquad= 0.
        \tag{by \cref{thm:lne-cex:mu3}}
\end{align*}
Thus, we have $(\alpha + (n-1)\mu_3 - \alpha\mu_3^2)(q_1-q) > 0$.
Moreover, since $\mu_3 > 1$ by \cref{thm:lne-cex:mu3},
we have $\alpha + (n-1)\mu_3 - \alpha\mu_3^2 = (n-1)(\mu_3 - 1)/q_1 > 0$.
Thus, $q < q_1$, which is a contradiction

Thus, our assumption that $\mu(p, q) < \mu_3$ was incorrect,
and we have $\mu(p, q) \ge \mu_3$.
\end{proof}

\subsection{\texorpdfstring{\boldmath}{}Case 2:
    \texorpdfstring{$p < \alpha q$}{p < aq}
    and \texorpdfstring{$p + q \ge 1$}{p + q >= 1}}
\label{sec:lne-cex:case2}

Recall \cref{thm:lne-cex:r,thm:lne-cex:rstar}.
Since $p < \alpha q$ and $p + q \ge 1$, we get
\begin{align*}
r_1^*(q) &= \max\left(\beta + (n-1)\max(0, 1-q), \min(\beta, \alpha q) + (n-1)\min(1, \alpha q)\right),
    \tag{by \cref{thm:lne-cex:rstar}}
\\ r_1(p, q) &= \min(\beta, p) + (n-1)\min(p, 1),
    \tag{by \cref{thm:lne-cex:r}}
\\ r_2^*(p) &= (n-1)\max(1-p, \min(1, p/\alpha)),
    \tag{by \cref{thm:lne-cex:rstar}}
\\ r_2(p, q) &= (n-1)\max(0, 1-p).
    \tag{by \cref{thm:lne-cex:r}, $q \ge 1 - p$}
\end{align*}

\begin{lemma}
\label{thm:lne-cex:2}
If $p < \alpha q$ and $p + q \ge 1$, we get $\mu(p, q) \ge \mu_1$.
\end{lemma}
\begin{proof}
If $p \ge 1$, then $r_2(p, q) = 0$, and if $p = 0$, then $r_1 = 0$.
In both cases, we get $\mu(p, q) = \infty$. So now let $p \in (0, 1)$.

By \cref{thm:lne-cex:p1} $p_1 \in (\frac{\alpha}{\alpha+1}, 1)$.

\textbf{Case 1}: $p \ge p_1$
\\ Then $p \ge \frac{\alpha}{\alpha+1}$, so $1-p \le p/\alpha$. Thus,
\[ \mu(p, q) \ge \frac{r_2^*(p)}{r_2(p, q)} = \frac{p/\alpha}{1-p}
    \ge \frac{p_1/\alpha}{1-p_1} = \frac{r_2^*(p_1)}{r_2(p_1, q_1)} = \mu_1. \]

\textbf{Case 2}: $p \le p_1$.
\\ By \cref{thm:lne-cex:q1}, $r_1^*(q) \ge r_1^*(q_1)$ for all $q \ge 0$. Thus,
\[ \mu(p, q) \ge \frac{r_1^*(q)}{r_1(p, q)} = \frac{r_1^*(q)}{np}
    \ge \frac{r_1^*(q_1)}{np_1} = \frac{r_1^*(q_1)}{r_1(p_1, q_1)} = \mu_1. \]

Hence, $\mu(p, q) \ge \mu_1$.
\end{proof}

\subsection{\texorpdfstring{\boldmath}{}Case 3:
    \texorpdfstring{$p + q \le 1$}{p + q <= 1}}
\label{sec:lne-cex:case3}

Recall \cref{thm:lne-cex:r,thm:lne-cex:rstar}.
Since $p + q \le 1$, we get
\begin{align*}
r_1^*(q) &= \max\left(\beta + (n-1)(1-q),
    \min(\beta, \alpha q) + (n-1)\min(1, \alpha q)\right),
    \tag{by \cref{thm:lne-cex:rstar}}
\\ r_1(p, q) &= \min(\beta, p) + (n-1)p,
    \tag{by \cref{thm:lne-cex:r}}
\\ r_2^*(p) &= (n-1)\max(1-p, p/\alpha),
    \tag{by \cref{thm:lne-cex:rstar}, $p \le 1$, $\alpha \ge 1$}
\\ r_2(p, q) &= (n-1)q.
    \tag{by \cref{thm:lne-cex:r}}
\end{align*}

Let $p' \defeq \frac{\alpha}{\alpha+1}$.
We first show that (informally) $\mu(p, q)$ is decreasing in $p$ if $p \le p'$
and $\mu(p, q)$ is decreasing in $q$ if $q \le q_1$.

\begin{lemma}
\label{thm:lne-cex:mu-dec-p}
Let $q \in [0, 1]$ and $0 \le x_1 \le x_2 \le \min(p', 1-q)$.
Then $\mu(x_1, q) \ge \mu(x_2, q)$.
\end{lemma}
\begin{proof}
For any $p \in [0, \min(p', 1-q)]$, we have $1-p \ge p/\alpha$, so
\[ \mu(p, q) = \max\left(\frac{r_1^*(q)}{r_1(p, q)}, \frac{r_2^*(p)}{r_2(p, q)}\right)
    = \max\left(\frac{r_1^*(q)}{np}, \frac{1-p}{q}\right). \]
Both terms in the $\max(\ldots)$ are decreasing in $p$.
\end{proof}

\begin{lemma}
\label{thm:lne-cex:mu-dec-q}
Let $p \in [0, 1]$ and $0 \le y_1 \le y_2 \le \min(q_1, 1-p)$.
Then $\mu(p, y_1) \ge \mu(p, y_2)$.
\end{lemma}
\begin{proof}
For any $q \in [0, 1]$, we have
$\beta + (n-1)(1-q) \ge \min(\beta, \alpha q) + (n-1)\min(1, \alpha q) \iff q \le q_1$.
Thus, for any $q \in [0, \min(q_1, 1-p)]$, we have
\[ \mu(p, q) = \max\left(\frac{r_1^*(q)}{r_1(p, q)}, \frac{r_2^*(p)}{r_2(p, q)}\right)
    = \max\left(\frac{\beta + (n-1)(1-q)}{np}, \frac{r_2^*(p)}{q}\right). \]
Both terms in the $\max(\ldots)$ are decreasing in $q$.
\end{proof}

\begin{lemma}
\label{thm:lne-cex:pq-dom}
For any $(p, q) \in \mathbb{R}_{\ge 0}^2$ such that $p + q \le 1$,
there exists $(\phat, \qhat) \in \mathbb{R}_{\ge 0}^2$ such that
$\phat + \qhat = 1$, $\phat \ge p$, $\qhat \ge q$,
and $\mu(\phat, \qhat) \le \mu(p, q)$.
\end{lemma}
\begin{proof}
Let $p' \defeq \frac{\alpha}{\alpha+1}$.
By \cref{thm:lne-cex:q1}, $q_1 \in (1/\alpha, 1)$.
Thus, $p' + q_1 > 1$.

Let $\phat \defeq \max(p, \min(p', 1-q))$
and $\qhat \defeq \max(q, \min(q_1, 1-\phat))$.
It is easy to see that $\phat \ge p$ and $\qhat \ge q$.

\textbf{Case 1}: $p \le 1-q \le p'$
\\ Then $\phat = 1-q$, and so, $\qhat = \max(q, \min(q_1, q)) = q$.
Then $\phat + \qhat = (1-q) + q = 1$.
By \cref{thm:lne-cex:mu-dec-p}, we get
$\mu(p, q) \ge \mu(\phat, q) = \mu(\phat, \qhat)$.

\textbf{Case 2}: $p \le p' \le 1-q$
\\ Then $\phat = p'$.
Since $p' + q_1 > 1$ and $p' \le 1-q$, we get $q \le 1-p' < q_1$.
Hence, $\qhat = \max(q, \min(q_1, 1-q_4)) = 1-p'$.
Thus, $\phat + \qhat = p' + (1-p') = 1$.
By \cref{thm:lne-cex:mu-dec-p}, we get $\mu(p, q) \ge \mu(\phat, q)$,
and by \cref{thm:lne-cex:mu-dec-q}, we get $\mu(\phat, q) \ge \mu(\phat, \qhat)$.
Thus, $\mu(p, q) \ge \mu(\phat, \qhat)$.

\textbf{Case 3}: $p' \le p \le 1-q$
\\ Then $\phat = p$.
Since $p' + q_1 > 1$, we get $q \le 1-p \le 1-p' < q_1$.
Thus, $\qhat = \max(q, \min(q_1, 1-p)) = 1-p$, so $\phat + \qhat = p + (1-p) = 1$.
Also, by \cref{thm:lne-cex:mu-dec-q}, $\mu(p, q) = \mu(\phat, q) \ge \mu(\phat, \qhat)$.

Thus, in each case, we get $\phat \ge p$, $\qhat \ge q$,
$\phat + \qhat = 1$, and $\mu(p, q) \ge \mu(\phat, \qhat)$.
\end{proof}

\begin{lemma}
\label{thm:lne-cex:3}
If $p + q \le 1$, then $\mu(p, q) \ge \min(\mu_1, \mu_2, \mu_3)$.
\end{lemma}
\begin{proof}
By \cref{thm:lne-cex:pq-dom}, there exist $\phat \ge p$ and $\qhat \ge q$
such that $\phat + \qhat = 1$ and $\mu(p, q) \ge \mu(\phat, \qhat)$.

If $\phat > \alpha q$, then by \cref{thm:lne-cex:1.1,thm:lne-cex:1.2},
we get $\mu(\phat, \qhat) \ge \min(\mu_2, \mu_3)$.
If $\phat < \alpha q$, then by \cref{thm:lne-cex:2}, we get
$\mu(\phat, \qhat) \ge \mu_1$.
Now let $\phat = \alpha q$. Then $\phat = \frac{\alpha}{\alpha+1}$ and $\qhat = \frac{1}{\alpha+1}$.

By \cref{thm:lne-cex:q1}, $\qhat < 1/\alpha < q_1$ and $r_1^*(\qhat) \ge r_1^*(q_1)$.
By \cref{thm:lne-cex:p1}, $\phat = < p_1$, and
\[ \mu(\phat, \qhat) \ge \frac{r_1^*(\qhat)}{r_1(\phat, \qhat)} = \frac{r_1^*(\qhat)}{n\phat}
    \ge \frac{r_1^*(q_1)}{np_1} = \frac{r_1^*(q_1)}{r_1(p_1, q_1)} = \mu_1. \]
Hence, $\mu(p, q) \ge \mu(\phat, \qhat) \ge \min(\mu_1, \mu_2, \mu_3)$.
\end{proof}

\subsection{\texorpdfstring{\boldmath}{}Case 4:
    \texorpdfstring{$p = \alpha q$}{p = aq}
    and \texorpdfstring{$p + q > 1$}{p + q > 1}}
\label{sec:lne-cex:case4}

\begin{observation}
$1-q < p = \alpha q \implies q > 1/(\alpha+1)$ and $p > \alpha/(\alpha+1)$.
\end{observation}

Since $p$ can be expressed as a function of $q$, we write $r_2^*(q)$ instead of $r_1^*(p)$.
By \cref{thm:lne-cex:rstar}, and using $p = \alpha q$ and $p + q > 1$, we get
\begin{align*}
r_1^*(q) &= \max\left(\beta + (n-1)\max(0, 1-q),
    \min(\beta, \alpha q) + (n-1)\min(1, \alpha q)\right),
\\ r_2^*(q) &= (n-1)\min(1, q),
\end{align*}

The revenue collected by the sellers cannot be precisely determined,
but we can still give upper and lower bounds on them.
For $j \in [2]$, define
\begin{enumerate}
\item $r_1^-(q) \defeq \min(\beta, \alpha q) + (n-1)\max(0, 1-q)$,
\item $r_1^+(q) \defeq \min(\beta, \alpha q) + (n-1)\min(1, \alpha q)$,
\item $r_2^-(q) \defeq (n-1)\max(0, 1-\alpha q)$,
\item $r_2^+(q) \defeq (n-1)\min(q, 1)$,
\item $s(q) \defeq \min(\beta, \alpha q) + (n-1)$.
\end{enumerate}
Then by \cref{thm:lne-cex:r}, the revenue of seller $j$ lies between
$r_j^-(q)$ and $r_j^+(q)$, and the sum of the sellers' revenues is $s(q)$.

\begin{observation}
For all $q > 1/(\alpha+1)$, we have $r_1^*(q) \ge r_1^+(q)$ and $r_2^*(q) = r_2^+(q)$.
\end{observation}

\begin{observation}
For all $q > 1/(\alpha+1)$, we have
$r_1^-(q) + r_2^+(q) = r_1^+(q) + r_2^-(q) = s(q)$.
Let $d(q) \defeq (n-1)(\min(q, 1) + \min(\alpha q, 1) - 1)$.
Then $r_1^+(q) - r_1^-(q) = r_2^+(q) - r_2^-(q) = d(q)$.
Moreover, $d(q) > 0$ for all $q > 1/(\alpha+1)$.
This is because if $q \ge 1/\alpha$, then $d(q) = (n-1)q > 0$,
and otherwise, $d(q) = (n-1)((\alpha+1)q - 1) > 0$.
\end{observation}

Seller 1's revenue will always be a convex combination of $r_1^-(q)$ and $r_1^+(q)$.
Denote her revenue by $r_1(q, z) \defeq (1-z) \cdot r_1^-(q) + z \cdot r_1^+(q)$.
Since the sellers' revenues sum to $s(q)$, we get that seller 2's revenue is
$r_2(q, z) \defeq z \cdot r_2^-(q) + (1-z) \cdot r_2^+(q)$.
We can rewrite these as, $r_1(q, z) = r_1^-(q) + z \cdot d(q)$
and $r_2(q, z) = r_2^+(q) - z \cdot d(q)$.

Let
\[ \mu(q, z) \defeq \max\left(\frac{r_1^*(q)}{r_1(q, z)}, \frac{r_2^*(q)}{r_2(q, z)}\right). \]
We would like to lower bound $\mu(q, z)$ for all $q > \frac{1}{\alpha+1}$ and $z \in [0, 1]$.

\begin{lemma}
\label{thm:lne-cex:4.1}
For all $q > \frac{1}{\alpha+1}$, let
\[ z^*(q) \defeq \min\left(1, \frac{r_2^*(q)(r_1^*(q) - r_1^-(q))}{(r_1^*(q) + r_2^*(q))d(q)}\right). \]
Then
\[ \inf_{z \in [0, 1]} \mu(q, z) = \mu(q, z^*(q))
    = \max\left(\frac{r_1^*(q)+r_2^*(q)}{s(q)}, \frac{r_1^*(q)}{r_1^+(q)}\right). \]
\end{lemma}
\begin{proof}
The key idea is that $r_1^*(q)/r_1(q, z)$ is decreasing in $z$
and $r_2^*(q)/r_2(q, z)$ is increasing in $z$.
If these curves (as functions of $z$) intersect,
then the point of intersection is $z^*$, and this minimizes $\mu(q, z)$.
Otherwise, we show that $\mu(q, z) = r_1^*(q)/r_1(q, z)$,
so $\mu(q, z)$ is minimized at $z = 1$.
We now formalize this argument.

For $j \in [2]$, let $g_j(q, z) \defeq r_j^*(q)/r_j(q, z)$. Let
\[ \zhat(q) \defeq \frac{r_2^*(q)(r_1^*(q) - r_1^-(q))}{(r_1^*(q) + r_2^*(q))d(q)}. \]
Then $\zhat(q) \ge 0$ and
\begin{align*}
g_1(q, z) \ge g_2(q, z)
    &\iff \frac{r_1^*(q)}{r_1(q, z)} \ge \frac{r_2^*(q)}{r_2(q, z)}
\\ &\iff r_1^*(q)r_2(q, z) \ge r_2^*(q)r_1(q, z)
\\ &\iff r_1^*(q)(r_2^*(q) - d(q)z) \ge r_2^*(q)(r_1^-(q) + d(q)z)
    \tag{since $r_2^+(q) = r_2^*(q)$}
\\ &\iff z \le \zhat(q).
\end{align*}
Similarly, we can show that $g_1(q, z) \le g_2(q, z) \iff z \ge \zhat(q)$.

\textbf{Case 1}: $\zhat(q) \ge 1$.
\\ Then $z^*(q) = 1$ and for all $z \in [0, 1]$, we have $g_1(q, z) \ge g_2(q, z)$.
Thus, $\mu(q, z) = g_1(q, z)$. Moreover, $g_1(q, z)$ is decreasing in $z$, so
\[ \inf_{z \in [0, 1]} \mu(q, z) = \inf_{z \in [0, 1]} g_1(q, z) = g_1(q, 1) = \mu(q, 1) = \mu(q, z^*(q)). \]
Moreover,
\[ \mu(q, z^*(q)) = g_1(q, 1) = \frac{r_1^*(q)}{r_1^+(q)}. \]

\textbf{Case 2}: $\zhat(q) \le 1$.
\\ Then $z^*(q) = \zhat(q)$.
For $z \le \zhat(q)$, we have $g_1(q, z) \ge g_2(q, z)$, so $\mu(q, z) = g_1(q, z)$.
Since $g_1(q, z)$ is decreasing in $z$, we have
$\mu(q, z) = g_1(q, z) \ge g_1(q, \zhat(q)) = \mu(q, \zhat(q))$.
For $z \ge \zhat(q)$, we have $g_1(q, z) \le g_2(q, z)$, so $\mu(q, z) = g_2(q, z)$.
Since $g_2(q, z)$ is increasing in $z$, we have
$\mu(q, z) = g_2(q, z) \ge g_2(q, \zhat(q)) = \mu(q, \zhat(q))$.
Thus, $\inf_{z \in [0, 1]} \mu(q, z) = \mu(q, z^*(q))$.
Moreover,
\[ \mu(q, z^*(q)) = g_2(z^*(q)) = \frac{r_2^*(q)}{r_2^*(q) - d(q)\zhat(q)}
    = \frac{1}{1 - \frac{r_1^*(q) - r_1^-(q)}{r_1^*(q)+r_2^*(q)}}
    = \frac{r_1^*(q) + r_2^*(q)}{r_2^*(q) + r_1^-(q)}. \]

On combining the two cases, we get
\[ \inf_{z \in [0, 1]} \mu(q, z) = \mu(q, z^*(q)) = \begin{cases}
    \frac{r_1^*(q)}{r_1^+(q)} & \text{ if } \zhat(q) \ge 1
    \\ \frac{r_1^*(q) + r_2^*(q)}{r_2^*(q) + r_1^-(q)} & \text{ otherwise}
\end{cases}. \]
\begin{align*}
& \frac{r_1^*(q)}{r_1^+(q)} \ge \frac{r_1^*(q) + r_2^*(q)}{r_2^*(q) + r_1^-(q)}
\\ &\iff r_1^*(q)(r_2^*(q) + r_1^-(q)) \ge (r_1^-(q) + d(q))(r_1^*(q) + r_2^*(q))
\\ &\iff r_1^*(q)r_2^*(q) \ge r_1^-(q)r_2^*(q) + d(q)(r_1^*(q) + r_2^*(q))
\\ &\iff \zhat(q) \ge 1.
\end{align*}
Thus,
\[ \inf_{z \in [0, 1]} \mu(q, z) = \mu(q, z^*(q))
    = \max\left(\frac{r_1^*(q)+r_2^*(q)}{r_2^*(q) + r_1^-(q)}, \frac{r_1^*(q)}{r_1^+(q)}\right).
    \qedhere \]
\end{proof}

\begin{lemma}
\label{thm:lne-cex:mu4}
\[ \mu(q_1, z^*(q_1)) = \mu_4 \defeq 1 + \frac{\beta(n-1)}{(n-1)^2 + (n-1)\alpha + \alpha\beta}. \]
\end{lemma}
\begin{proof}
By \cref{thm:lne-cex:q1}, we get $q_1 \in (1/\alpha, 1)$
and $r_1^*(q_1) = n\chat_1 = \alpha q_1 + n-1$.
Also, $r_2^*(q_1) = r_2^+(q) = (n-1)q_1$,
$r_1^+(q_1) = \alpha q_1 + (n-1)$, and $r_1^-(q_1) = \alpha q_1 + (n-1)(1-q_1)$.
Thus, by \cref{thm:lne-cex:4.1}, we get
\begin{align*}
\mu(q_1, z^*(q_1))
    &= \max\left(\frac{r_1^*(q_1)}{r_1^+(q_1)}, \frac{r_1^*(q_1)+r_2^*(q_1)}{r_1^-(q)+r_2^*(q)}\right)
    = \max\left(1, \frac{\alpha q_1 + (n-1)(1+q_1)}{\alpha q_1 + n-1}\right)
\\ &= 1 + \frac{(n-1)q_1}{\alpha q_1 + n-1}
    = 1 + \frac{(n-1)\beta}{(n-1)^2 + (n-1)\alpha + \alpha\beta}.
\qedhere
\end{align*}
\end{proof}

\begin{lemma}
\label{thm:lne-cex:4.2}
For all $q > \frac{1}{\alpha+1}$, we have
\[ \mu(q, z^*(q)) \ge \min(\mu_2, \mu_4). \]
\end{lemma}
\begin{proof}
Let $h(q) \defeq \beta + (n-1)\max(0, 1-q)$.
Then $r_1^*(q) = \max(h_1(q), r_1^+(q))$.
Using techniques from the proof of \cref{thm:lne-cex:q1},
one can easily show that $h_1(q) \ge r_1^+(q) \iff q \le q_1$
and $h_1(q) \le r_1^+(q) \iff q \ge q_1$.

\textbf{Case 1}: $q \in (\frac{1}{\alpha+1}, q_1]$. Then
\[ \frac{r_1^*(q)}{r_1^+(q)} = \max\left(1, \frac{h_1(q)}{r_1^+(q)}\right). \]
Since $h_1(q)$ is non-increasing in $q$ and $r_1^+(q)$ is non-decreasing in $q$,
we get that $r_1^*(q)/r_1^+(q)$ is non-increasing in $q$.
\[ \frac{r_1^*(q)+r_2^*(q)}{s(q)} = \frac{\beta + n-1}{s(q)}. \]
Since $s(q)$ is non-decreasing in $q$, we get that $(r_1^*(q)+r_2^*(q))/s(q)$
is non-increasing in $q$. Thus,
$\mu(q, z^*(q))$ is non-increasing in $q$.
So, $\mu(q, z^*(q)) \ge \mu(q_1, z^*(q_1)) = \mu_4$
by \cref{thm:lne-cex:mu4}.

\textbf{Case 2}: $q \ge q_1$.
\\ Then $q_1 \ge 1/\alpha$ by \cref{thm:lne-cex:q1}.
Also, $r_1^*(q)/r_1^+(q) = 1$.
Since $r_1^*(q) + r_2^*(q) = r_1^+(q) + r_2^+(q) = s(q) + d(q)$, we get
\[ \frac{r_1^*(q)+r_2^*(q)}{s(q)} = 1 + \frac{d(q)}{s(q)}
= 1 + \frac{(n-1)\min(q, 1)}{\min(\beta, \alpha q) + (n-1)}. \]

\textbf{Case 2a}: $q \in [q_1, 1]$. Then
\[ \frac{r_1^*(q)+r_2^*(q)}{s(q)} = 1 + \frac{(n-1)q}{\alpha q + n-1}. \]
One can show that this is increasing in $q$.
Thus, for $q \in [q_1, 1]$, we have $\mu(q, z^*(q)) \ge \mu(q_1, z^*(q_1)) = \mu_4$
by \cref{thm:lne-cex:mu4}.

\textbf{Case 2b}: $q \ge 1$. Then
\[ \frac{r_1^*(q)+r_2^*(q)}{s(q)} = 1 + \frac{n-1}{\min(\beta, \alpha q) + (n-1)}, \]
which is non-increasing in $q$.
Thus, for $q \ge 1$, we get
\[ \mu(q, z^*(q)) \ge \mu(\infty, z^*(\infty)) = 1 + \frac{n-1}{\beta + n-1}. \]
To complete the proof, we need to show that
\[ 1 + \frac{n-1}{\beta+n-1} \ge \mu_2. \]
This holds because
\begin{align*}
\mu_2 &= \frac{\sqrt{(n-1)^2+(n-1)\alpha+\alpha\beta}+(n-1)}{\beta + n-1}
\\ &\le \frac{\sqrt{(n-1)^2+(n-1)\alpha+\alpha(\alpha + n-1)}+(n-1)}{\beta + n-1}
    \tag{since $\beta \le \alpha + n-1$ by Constraint \ref{item:lne-cex:c1}}
\\ &= \frac{\alpha+2(n-1)}{\beta + n-1}
    \le \frac{\beta+2(n-1)}{\beta + n-1}
    = 1 + \frac{n-1}{\beta + n-1}.
\end{align*}

Thus, in all cases, we get that $\mu(q, z^*(q)) \ge \min(\mu_2, \mu_4)$.
\end{proof}

\section{Details on Non-Existence of NE for Convex Pricing}
\label{sec:plc-cex-extra}

We begin by investigating buyers' behavior for convex pricing functions.
A useful tool for this is the sub-derivative.

\begin{definition}[Sub-derivative]
\label{defn:sub-derivative}
Let $p: [0, 1] \to \mathbb{R}$ be a continuous convex function.
$\alpha \in \mathbb{R}$ is said to be a sub-derivative of $p$ at $z \in [0, 1]$ if
for all $x \in [0, 1]$, we have $\alpha\cdot(z - x) \le p(x) - p(z)$.
\end{definition}

Sub-derivatives of a continuous convex function $p: [0, 1] \to \mathbb{R}$ have two useful properties:
\begin{itemize}
\item \textbf{Surjectivity}: For any $\alpha \in \mathbb{R}$, $\exists x \in [0, 1]$ such that
    $\alpha$ is a sub-derivative of $p$ at $x$.
\item \textbf{Monotonicity}: Let $x_1 < x_2$, $\alpha_1$ be a sub-derivative of $p$ at $x_1$,
    and $\alpha_2$ be a sub-derivative of $p$ at $x_2$. Then $\alpha_1 \le \alpha_2$.
\end{itemize}

Consider two datasets having convex pricing functions $p_1$ and $p_2$, respectively.
Buyers would then have to decide how much of each dataset to purchase.
Formally, a buyer $i$ having budget $b_i$ and utilities $\tau_{i,1}$ and $\tau_{i,2}$
for the datasets would face the following optimization problem:
\[ \maximize_{x_1, x_2 \in [0, 1]} \; \tau_{i,1}x_1 + \tau_{i,2}x_2
    \text{ subject to } p_1(x_1) + p_2(x_2) \le b_i. \]
Then seller $j \in [2]$ earns revenue $p_j(x_j)$ from buyer $i$.
There may be multiple optimal solutions, and we say that a buyer is \emph{biased} towards seller $j$
if among the optimal solutions, she picks the one maximizing $x_j$.

In \cref{sec:notation:revenue}, we showed that for linear pricing,
the above optimization problem is the fractional knapsack problem,
and buyers solve it by ordering the datasets in bang-per-buck order and purchasing greedily.
A similar observation can be made for convex prices too.
For two convex-priced datasets, a solution $(x_1, x_2)$ is optimal iff the following hold:
\begin{itemize}
\item $p_1(x_1) + p_2(x_2) = \min(b_i, p_1(1) + p_1(2))$.
\item $\alpha_1/\tau_{i,1} = \alpha_2/\tau_{i,2}$, where $\alpha_j \in \mathbb{R}_{\ge 0}$
    is a sub-derivative of $p_j$ at $x_j$ for all $j \in [2]$.
\end{itemize}

We are now ready to prove the main theorem.

\thmConvexNeCex*
\begin{proof}
The total budget among the buyers is $(n-1) + (n+1) = 2n$.
For any pricing function selected by a seller, we show that the other seller
can pick a pricing function that gives her revenue more than $n$.
This would prove that a Nash equilibrium doesn't exist,
because in every pair of pricing functions, some seller will always earn at most $n$,
and that seller can profitably deviate to earn more than $n$.

By symmetry, we only need to prove one side: given seller 1's pricing function $p_1$,
we need to find a pricing function $p_2$ such that seller 2's revenue is strictly more than $n$.

Assume \wLoG{} that buyers are biased towards seller 2,
since she can get nearly the same revenue as the biased case
by reducing her pricing function by a factor $(1-\eps)$
for an infinitesimally small $\eps > 0$.

For all $j \in [2]$, let $r_j(p_1, p_2)$ denote seller $j$'s revenue for pricing profile $(p_1, p_2)$
(when buyers are biased towards seller 2).
If $r_2(p_1, p_1) > n$, then we are done. Now assume $r_2(p_1, p_1) \le n$.

\paragraph{Case 1:} $p_1(1) < (n+1)/2$.
\\ If $p_1(1) < 1$, then seller 1's revenue is less than $n$.
If we define $p_2(x) \defeq (n+1)x$, then $r_2(p_1, p_2) > n$, and we are done.
Now assume $p_1(1) \ge 1$.

When prices are $(p_1, p_1)$, the poor buyers exhaust their budgets, and the rich buyer has
budget $\mu \defeq (n+1) - 2p_1(1)$ left after purchasing both datasets.
Let $\eta$ be the total revenue earned by seller 2 from poor buyers for prices $(p_1, p_1)$.
Her revenue from the rich buyer is $p_1(1)$. Thus, $r_1(p_1, p_1) \le r_2(p_1, p_1) = \eta + p_1(1)$.
Moreover, $r_1(p_1, p_1) + r_2(p_1, p_1)$ equals the total expenditure $2n - \mu$,
so $r_2(p_1, p_1) \ge n - \mu/2$.

Pick $\delta$ in $(0, \mu/(2\eta))$.
We obtain $p_2$ by shrinking $p_1$ along the $x$ and $y$ axes by a factor of $1-\delta$,
and setting $p_2$'s derivative in the interval $[1-\delta, 1]$ to a large value
so that $p_2(1) = p_1(1) + \mu$. Formally, define $p_2$ as follows:
\[ p_2(x) \defeq \begin{cases}
(1-\delta)p_1(x/(1-\delta)) & \text{ if } x \le 1-\delta
\\ p_1(1) \cdot x + \mu \cdot \frac{x - (1-\delta)}{\delta} & \text{ if } x > 1-\delta
\end{cases}. \]
Thus, when prices are $(p_1, p_2)$,
the rich buyer purchases both datasets fully and exhausts her budget.

The total revenue earned by seller 2 from poor buyers for prices $(p_1, p_2)$ is at least
$(1-\delta)\eta$. Her revenue from the rich buyer is
$p_2(1) = p_1(1) + \mu$. Thus,
\begin{align*}
r_2(p_1, p_2) &\ge (1-\delta)\eta + (p_1(1) + \mu) = r_2(p_1, p_1) + (\mu - \eta\delta)
\\ &> (n - \mu/2) + \mu/2 = n.
\end{align*}

\paragraph{Case 2:} $p_1(1) \ge (n+1)/2$.
\\ Then every buyer exhausts her budget at prices $(p_1, p_1)$, so $r_1(p_1, p_1) + r_2(p_1, p_1) = 2n$.
Since $r_1(p_1, p_1) \le r_2(p_1, p_1) \le n$, we get $r_1(p_1, p_1) = r_2(p_1, p_1) = n$.

Let $\alpha$ be the smallest number in $(0, 1)$ such that $p_1(\alpha) + \alpha \cdot p_1'(\alpha) = 1$,
for some sub-derivative $p_1'(\alpha)$ of $p_1$ at $\alpha$.
Such an $\alpha$ and $p_1'(\alpha)$ always exist, since $p_1(1) \ge (n+1)/2 > 1$, $0 \cdot p_1'(0) = 0 < 1$,
and $p_1(x) + x \cdot p_1'(x)$ is monotonic in $x$.

Let $\beta$ be the smallest number in $(\alpha, 1]$ such that
$p_1(\beta) - p_1(\alpha) + (\beta - \alpha) \cdot p_1'(\beta) = n$,
for some sub-derivative $p_1'(\beta)$ of $p_1$ at $\beta$.
Such a $\beta$ and $p_1'(\beta)$ always exist, since $p_1'(1)$ can be made arbitrarily large.
Define $p_2$ as
\[ p_2(x) \defeq \begin{cases}
p_1'(\alpha) \cdot x & \text{ if } 0 \le x \le \alpha
\\ \alpha \cdot p_1'(\alpha) + p_1'(\beta) \cdot (x - \alpha) & \text{ if } \alpha < x \le 1
\end{cases}. \]
Then the poor buyers will purchase an $\alpha$ fraction of each dataset,
and the rich buyer will purchase at least a $\beta$ fraction of dataset 2.
All buyers will exhaust their budgets.
\[ r_2(p_1, p_2) - r_1(p_1, p_2)
\ge n(\alpha \cdot p_1'(\alpha) - p_1(\alpha)) + ((\beta - \alpha) \cdot p_1'(\beta) - (p_1(\beta) - p_1(\alpha))) \]
Both these terms are non-negative by convexity of $p_1$.
If any of them are positive, we get $r_2(p_1, p_2) > r_1(p_1, p_2)$, which implies $r_2(p_1, p_2) > n$
(since the total expenditure is $2n$) and we are done.
Now assume that these terms are 0. This implies that $p_1(x) = p_2(x)$ for all $x < \beta$.

We get $\alpha \cdot p_1'(\alpha) = 1/2$ and $(\beta - \alpha) \cdot p_1'(\beta) = n/2$.
We price dataset 2 at $p_3(x) = p_1'(\beta)x$.
If $p_1'(\alpha) = p_1'(\beta)$, then $r_2(p_1, p_3) = (n-1) + (n+1)/2 > n$.
Otherwise, $r_2(p_1, p_3) = (n-1)/2 + \min(n+1/2, p_1'(\beta))$.
Let $\nu \defeq \alpha \cdot (p_1'(\beta) - p_1'(\alpha)) > 0$. Then
\[ p_1'(\beta) \ge \beta \cdot p_1'(\beta)
= \nu + \alpha \cdot p_1'(\alpha) + (\beta - \alpha) p_1'(\beta) = \nu + (n+1)/2, \]
Thus, $r_2(p_1, p_3) \ge n + \min(\nu, n/2) > n$.
This concludes the proof.
\end{proof}

\section{NP-Hardness for Finding the Best Response}
\label{sec:np-hardness-br}

In this section, we show that it is \NP-hard for a seller to find their best response to other sellers' strategies $\ell_{-j}$.
Fix other sellers' strategies.
Denote by $\gamma_{i,s}$ the remaining budgets of buyer $i$ before the $s$-th shard of seller $j$, assuming that the buyer does not spend any money on any seller $j$'s data.
Let $\bm{\ell}_j= (\ell_{j,s})_{s\in [m_j]}$ be seller $j$'s strategy.
Then buyer $i$'s remaining budget before shard $j$ is given by $\max(0, r_{i,s} - \sum_{s' < s} p_{j, s'}\cdot \ell_{j, s'})$.
Therefore, the optimal strategy for seller $j$ can be characterized by the following program:
\begin{equation}
\begin{aligned}
\max_{\bell, \z}\quad  & \sum_{i=1}^n \sum_{s=1}^{m_j}  z_{i,s}\\
\text{subject to} \quad & z_{i,s} \le p_{j, s}\cdot \ell_{s} \\
& z_{i, s} \le \max\left(0, \gamma_{i, s}- \sum_{s'<s} p_{j, s'} \cdot \ell_{s'}\right) \\
& \sum_{s=1}^{m_j} \ell_s = 1; \quad  \ell_s \ge 0
\end{aligned}
\label{eq:best-response}
\end{equation}
We note that the above program is non-convex due to the second constraint---the right-hand side function $\max(0, \gamma_{i, s}- \sum_{s'<s} p_{j, s'} \cdot \ell_{s'})$ is non-concave---which makes the feasible domain non-convex.
The non-convexity of the feasible domain further leads to the hardness.
Below, we perform an equivalent transformation of the above program.
It can be observed that the optimal value is achieved when $z_{i,s} = \min(p_{j,s}\ell_s, \gamma_{i, s} - \max(0, \sum_{s'<s} p_{j, s'}\ell_{s'}))$.
Hence, the original program~(\ref{eq:best-response}) can be reformulated as
\begin{align*}
\max_{\bell}\quad  & \sum_{i=1}^n\sum_{s=1}^{m_j}\min \left(p_{j,s}\ell_s, \max\left(0, \gamma_{i, s} - \sum_{s'<s} p_{j, s'}\ell_{s'}\right)\right)\\
\text{subject to} \quad &\bell\in \Delta_{m_j}
\end{align*}
where $\bm{\ell} = (\ell_1, \dots, \ell_{m_j})$ and $\Delta_{m_j}$ is the $m_j$-simplex.
Let $L_s = \sum_{d=1}^{s} p_{j,s} \cdot \ell_{d}$ and $L_{0} = 0$.
Let $h_{i,s}(t) = \boolOne[\gamma_{i,s} \ge t]$.
Then the inner term of the objective value can be rewritten as
$$
\min \left(p_{j,s}\ell_s, \max\left(0, \gamma_{i, s} - \sum_{s'<s} p_{j, s'}\ell_{s'}\right)\right)  = \int_{L_{s-1}}^{L_s} h_{i, s}(t)dt
$$
Therefore, we can rewrite the original program as follows:
\begin{align*}
\sum_{i=1}^n\sum_{s=1}^{m_j} \int_{L_{s-1}}^{L_s} h_{i, s}(t)dt &= \sum_{s=1}^{m_j} \int_{L_{s-1}}^{L_s} \sum_{i=1}^n h_{i,s}(t)dt \\
& \triangleq \sum_{s=1}^{m_j} \int_{L_{s-1}}^{L_s}  h_{s}(t)dt  \tag{let $h_s(t) = \sum_{i=1}^n h_{i,s}(t)$} \\
& = \sum_{s=1}^{m_s} \int_{0}^{L_s} h_s(t) dt - \int_{0}^{L_{s-1}} h_s(t) dt \\
& = \sum_{s=1}^{m_s} H_s(L_s) - H_{s}(L_{s-1}) \tag{$H_s(x) = \int_{0}^xh_s(t)dt$}\\
& = \sum_{s=1}^{m_j} \Phi_s(L_s) \tag{define $\Phi_s = H_s - H_{s+1}$ if $s<m_j$ and $\Phi_{m_j} = H_{m_j}$}
\end{align*}
Note that, $\bell \ge \bm{0}$ is equivalent to $0 \le L_{1} \le \dots \le L_{m_j}$.
Meanwhile, $\norm{\bell}_1 = 1$ is equivalent to
\begin{align*}
\sum_{s=1}^{m_j} \frac{L_s - L_{s-1}}{p_s} = 1 \iff  \sum_{s=1}^{m_j-1} L_s\cdot \left(\frac{1}{p_s} - \frac{1}{p_{s+1}}\right) + \frac{L_{m_j}}{p_{m_j}} = 1
\end{align*}
Let $\beta_s = 1/p_s - 1/p_{s+1}$ for any $s\in [m_j - 1]$ and $\beta_s = 1/p_{m_j}$ for $s = m_j$.
Each function $\Phi_s(\cdot)$ is piecewise-linear and monotone increasing.
For $s\in [m_j-1]$,
\begin{align}
\Phi_s(x) &= H_s(x) - H_{s+1}(x) \nonumber\\
& = \int_0^x h_s(t) - h_{s+1}(t)dt  \tag{$H_s(x) = \int_{0}^x h_s(t)dt$} \\
& = \int_0^x \sum_{i=1}^n \boolOne\left[(\gamma_{i,s} \ge t )\wedge (\gamma_{i, s+1}\le t)\right]dt \tag{as $\gamma_{i,s} \ge \gamma_{i, s+1}$} \\
& = \sum_{i=1}^n \min\left(\gamma_{i, s} - \gamma_{i, s+1}, \max(0, x- \gamma_{i, s+1})\right), \label{eq:phi-s}
\end{align}
Additionally, for $s=m_j$,
\begin{align*}
\Phi_{m_j}(x) & = H_{m_j}(x) = \int_0^x h_{m_j}(t) dt \\
& = \sum_{i=1}^n \int_0^x \boolOne[\gamma_{i, m_j} \ge t] dt  = \sum_{i=1}^n \min(\gamma_{i, m_j}, x)\,.
\end{align*}
Since each $\Phi_s$ is a sum of piecewise-linear, monotone functions, $\Phi_s$ is itself piecewise-linear and monotone.
Due to the monotonicity of $\Phi_s$, we can also relax the equality constraint as $\sum_s L_s\beta_s\le 1$.
The original program~(\ref{eq:best-response}) can be reformulated to the following equivalent program:
\begin{equation}
\begin{aligned}
\max_{\bm{L}}\quad
& \sum_{s=1}^{m_j} \Phi_s (L_s) \\
\text{subject to}\quad
& \sum_{s=1}^{m_j} L_s\cdot \beta_s \le 1 \\
& 0 \le L_1 \le \cdots \le L_{m_j}.
\end{aligned}
\label{eq:best-response-reformulated}
\end{equation}

However, the presence of both min and max operators in $\Phi_s$ implies that it is neither convex nor concave.
Below we show that solving program~(\ref{eq:best-response-reformulated}) is \NP-hard.

\begin{problem}
\textup{\textsc{(X3C)}}
Given a universe $U$ of $3q$ elements and a collection $\mathcal{S}$ of subsets of $U$, each of size exactly $3$.
Output `yes' if there exists an exact cover for $U$, i.e., a sub-collection $\mathcal{S}' \subseteq \mathcal{S}$ such that every element in $U$ occurs in exactly one member of $\mathcal{S}'$ and $\abs{\mathcal{S}'} = q$, and `no' otherwise.
\end{problem}

\begin{theorem}
Finding the best response pricing strategy given other sellers' strategies is \NP-hard.
\end{theorem}
\begin{proof}
By the above calculation, we know that it is equivalent to prove the \NP-hardness of solving program~(\ref{eq:best-response-reformulated}).
Below, we show a reduction from \xc to program~(\ref{eq:best-response-reformulated}).
Given a \xc instance, without loss of generality, we assume that the elements in $U$ are positive integers $1, \dots, \abs{U}$.
Let $m$ be the number of sets in $\mathcal{S}$.
For every set $S_i\in \mathcal{S}$, we create a weight $w_i = \sum_{u\in S_i}B^{-u}$.
Let $W$ be the greatest common multiple of the numerators of all weights.
Without loss of generality, we assume $W$ is an even number.
Then we normalize the weights by setting $w_i$  as $w_i / W$.
Associate item $i$ with a value of $v_i = \frac12 + w_i$.
Next, we construct a best-response pricing instance as follows.
\begin{itemize}[leftmargin=*]
\item The number of shards $m_j$ is equal to $q+1$.
\item Let $W_{X3C} = \sum_{i\in U} \frac{B^{-i}}{W}$.
We set the prices as follows.
Let $p_{j, q+1}= \frac{1+ W\cdot B^{\abs{U}}}{\Delta}$ and $p_{j, s} = \frac{p_{j, s+1}}{1+ \Delta\cdot p_{j, s+1}}$.
Hence, $\beta_1 = \dots = \beta_q = \Delta$ and $\beta_{q+1} = \Delta/p_{j, q+1}$, where $\Delta = (1/p_{j,q+1} + W_{X3C})^{-1}$.

\item We construct $q$ groups of buyers $\G^1, \dots, \G^q$ and each group $\G^s$ consists of $m$ subgroups of buyers $\G_1^s,\dots, \G_m^s$.
We refer to the outer groups as \emph{shard groups} and the inner groups as \emph{item groups}.
For any buyer $i$ in the $s$-th shard group $\G^s$, we set
$$\gamma_{i, 1} = \dots = \gamma_{i, s} > \gamma_{i, s+1} = \gamma_{i, s+2} = \dots = \gamma_{i, q+1},
$$
where there is only a gap between $\gamma_{i, s}$ and $\gamma_{i, s+1}$.
The $r$-th item group $\G_r^s$ contains $\frac{v_r}{w_r} = \frac{1}{2w_r} + 1$ buyers.
For any buyer $i$ in the $r$-th group, we set $\gamma_{i, s} = w_r$ and $\gamma_{i, s+1} = w_r - \frac{1}{\abs{\G_r^s}}\cdot (w_r - w_{r-1}) = w_r (1 -2\frac{w_r - w_{r-1}}{1 + 2w_r})$, where $w_0 = 0$.
Denote $w_r (1 -2\frac{w_r - w_{r-1}}{1 + 2w_r})$ by $w_r^-$.
Therefore, we can assume every $L_s^*$ lies in some interval $(w_r^-, w_r]$ for some $r\in [m]$.
\item We construct a dummy buyer with $\gamma_{i, 1} = \dots = \gamma_{i, q+1} = 1$.
Let $n$ be the number of buyers, and the dummy buyer be the $n$-th buyer.
\end{itemize}

Next, we take a further look at the form of the function $\Phi_s$ in the constructed instance.
For $s\in [q]$, according to Equation~(\ref{eq:phi-s}), the value of $\Phi_s(x)$ depends only on the values of $\gamma_{i,s}$ and $\gamma_{i, s+1}$.
According to our construction of shard groups, only the buyers in the $s$-th shard group $\G^s$ have gaps between $\gamma_{i,s}$ and $\gamma_{i, s+1}$.
The buyers with equal valued $\gamma_{i,s}$ and $\gamma_{i, s+1}$ does affect the value of $\Phi_s(x)$.
Within the shard group $\G^s$, the values of $\gamma_{i,s}$ and $\gamma_{i, s+1}$ are determined by the item group that buyer $i$ belongs to.
The buyers in the $r$-th item group $\G_r^s$ have $\gamma_{i,s} = w_r$ and $\gamma_{i, s+1} = w_r^-$.
It can be verified that $(w_r^-, w_r]$ does not overlap with each other as the weights $w_i$ are set small enough.
By the expression of $\Phi_s(x)$ in~(\ref{eq:phi-s}), the derivative on each interval $(w_r^-, w_r]$ is given by the number of buyers in item groups $\G_{r}^s$.
Therefore, we have $\Phi_s(x)$ is piecewise linear with breakpoints at $w_1^-, w_1, \dots, w_m^-, w_m$ and closed-form as follows:
\begin{equation}
\begin{aligned}
\Phi_s(x) =
    \begin{cases}
    \abs{\G_r^s}\cdot (x - w_r^-) + \frac12 + w_{r-1} & \text{if } x \in (w_r^-, w_r] \\
    \frac12 + w_r & \text{if } x \in (w_{r-1}, w_r^-]
    \end{cases}
\end{aligned}
\label{eq:phi_s_q}
\end{equation}
where we let $w_0 = 0$ and $r\le m$.
An illustration of $\Phi_s(x)$ for $s\in [q]$ is given in Fig.~\ref{fig:phi-s}.

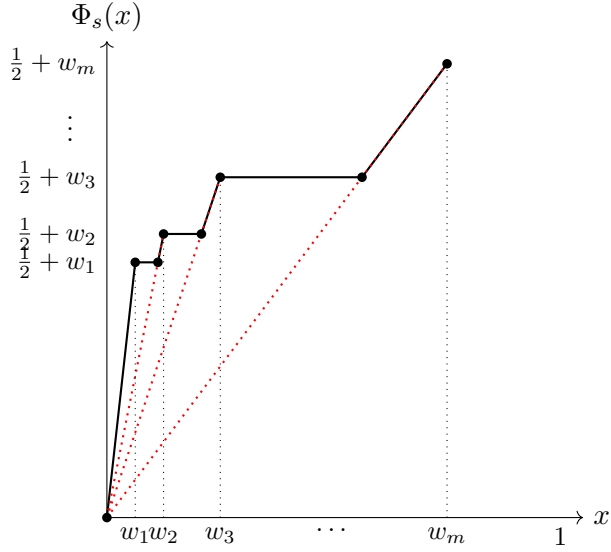
\begin{figure}[t]
\centering
\begin{tikzpicture}[scale=6]
    \draw[->] (0,0) -- (1.05,0) node[right] {$x$};
    \draw[->] (0,0) -- (0,1.05) node[above] {$\Phi_s(x)$};

    \node at (0, 9/16) [left] {\small $\frac12 + w_1$};
    \node at (1/16, 0) [below] {\small $w_1$};
    \node at (0, 5/8) [left] {\small $\frac12 + w_2$};
    \node at (1/8, 0) [below] {\small $w_2$};
    \node at (0, 3/4) [left] {\small $\frac12 + w_3$};
    \node at (1/4, 0) [below] {\small $w_3$};
    \node at (1/2, 0) [below] {\small $\dots$};
    \node at (0, 1) [left] {\small $\frac12 + w_m$};
    \node at (-0.05, 7/8) [left] {\small $\vdots$};
    \node at (3/4, 0) [below] {\small $w_m$};
    \node at (1, 0) [below] {\small $1$};

  \draw[thick]
    (0,0)
    -- (1/16, 9/16)
    -- (9/80, 9/16)
    -- (1/8, 5/8)
    -- (5/24, 5/8)
    -- (1/4, 3/4)
    -- (9/16, 3/4)
    -- (3/4, 1);

    \draw[dotted, thick, red] (0,0) -- ({1/8},{5/8});
    \draw[dotted] (1/16,0) -- (1/16,9/16);
    \draw[dotted, thick, red] (0,0) -- ({1/4},{3/4});
    \draw[dotted] (1/8,0) -- (1/8,5/8);
    \draw[dotted, thick, red] (0,0) -- ({3/4},{1});
    \draw[dotted] (1/4,0) -- (1/4,3/4);
    \draw[dotted] (3/4,0) -- (3/4,1);

    \foreach \x/\y in {
    {0}/{0},
    {1/16}/{9/16},
    {9/80}/{9/16},
    {1/8}/{5/8},
    {5/24}/{5/8},
    {1/4}/{3/4},
    {9/16}/{3/4},
    {3/4}/{1}
    } {
    \fill (\x,\y) circle (0.3pt);
    }
\end{tikzpicture}
\caption{Illustration of function $\Phi_s(x)$}
\label{fig:phi-s}
\end{figure}

In addition, for $s=q+1$, since the dummy buyer has  $\gamma_{n, q+1} = 1$,  we have
\begin{align*}
\Phi_{q+1}(x) = \sum_{i=1}^n \min(\gamma_{i, q+1}, x) = \sum_{i=1}^{n-1} \min(\gamma_{i, q+1}, x) + \min(1, x) \end{align*}
Hence, it can be observed that the derivative of $\Phi_{q+1}(x)$ is at least $1$ when $x< 1$.

\begin{lemma}[Near-Integrality]
\label{lem:almost-integeral}
Let $\bm{L}^* = (L_1^*, \dots, L_{q+1}^*)$ be an optimal solution to Program~(\ref{eq:best-response-reformulated}).
Then, one can construct another optimal solution $\bm{L}^\dag = (L_1^\dag, \dots, L_{q+1}^\dag)$ such that 1) at least $q-1$ of $L_1^\dag, \dots, L_q^\dag$ are in the set $\{w_1, \dots, w_m\}$; 2) at most one of $L_1^\dag, \dots, L_q^\dag$ is not in the set $\{w_1, \dots, w_m\}$ but still located in some interval $(w_r^-, w_r)$; 3) $L_{q+1}^\dag \ge 1$.
\end{lemma}
\begin{proof}
We first claim that, in the optimal solution $\bm{L}^*$, $L_{q+1}^*$ must be no less than one.
Otherwise, assume $L_{q+1}^* < 1$.
Let $L_s^*$ be the smallest positive number among $L_1^*, \dots, L_q^*$.
We can increase $L_{q+1}^*$ by amount of $\frac{\delta}{\beta_{q+1}}$ and decrease $L_s^*$ by the amount of $\frac{\delta}{\beta_s}$ with $L_{q+1}^*$ still no more than $1$.
The change of the objective value is at least
\begin{align*}
\frac{\delta}{\beta_{q+1}} - \frac{\delta}{\beta_s} \cdot \max_{r}\abs{\G_r^s} &= \delta\cdot\Delta \left(1+ W\cdot B^{\abs{U}} - \max_{r}\abs{\G_r^s}\right) \\
&= \delta\cdot\Delta \Big(1+ W\cdot B^{\abs{U}} - \max_{r} \Big(1+ \frac{1}{2w_r}\Big)\Big) \\
& = \delta\cdot\Delta \Big(W\cdot B^{\abs{U}} - \frac{W}{2\min_r\sum_{u\in S_r}B^{-u}}\Big)  > 0,
\end{align*}
which contradicts the optimality of $\bm{L}^*$.
Therefore, $L_{q+1}^*\ge 1$.

Next, we change the optimal solution $\bm{L}^*$ step by step.
It can be observed that all the $\Phi_s(x)$ for $s\in [q]$ have the same form~(\ref{eq:phi_s_q}).
Also $\beta_s$ is identical to $\Delta$ for all $s\in [q]$.
Therefore, we do not force the order of $L_1^*, \dots, L_q^*$ to be maintained during the transformation (otherwise, one can sort and relabel them).

First, if there exist some $L_s^* \in (w_r, w_{r+1}^-]$ for some $r\in [m-1]$ and $s\in [q]$, then we decrease it to $w_r$ which still maintains the linear constraint $\sum_s \beta_s L_s^* \le 1$.
Since the derivative of $\Phi_s(x)$ is zero in $(w_r, w_{r+1}^-]$, the objective value does not decrease, and the order of $L_s^*$'s is still maintained.

If there are more than two $L_s^*$ and $L_t^*$ not located at some $w_r$ for $r\in [m]$ and $s, t\in [q]$, then we proceed as follows.
Suppose $L_s^* \in (w_{r(s)}^-, w_{r(s)}]$ and $L_t^* \in (w_{r(t)}^-, w_{r(t)}]$ for some $r(s), r(t) \in [m]$.
If $r(s) \neq r(t)$, without loss of generality, assume $r(s) < r(t)$.
Then the derivative of $(w_{r(s)}^-, w_{r(s)}]$ is higher than that of $(w_{r(t)}^-, w_{r(t)}]$.
We increase $L_s^*$ and decrease $L_t^*$ by the same amount until one of them reaches the boundary of the interval.
If $r(s) = r(t)$, then without loss of generality, assume $L_s^* < L_t^*$.
Then we decrease $L_s^*$ and increase $L_t^*$ by the same amount until one of them reaches the boundary of the interval.
In either case, the objective value does not decrease.
We can repeat the above procedure and the first step until at most one of $L_1^*, \dots, L_q^*$ is not located at some $w_r$ for $r\in [m]$.

As the above process does not decrease $L_{q+1}^*$, we have constructed another optimal solution $\bm{L}^\dag$ satisfying all the desired properties.
\end{proof}

Lastly, we prove the correctness of the reduction.
If the \xc instance is a yes instance, then we show that the optimal value of program~(\ref{eq:best-response-reformulated}) is at least $q/2 + W_{X3c} + \sum_{i=1}^n \gamma_{i, q+1}$.
Let $\mathcal{S}'$ be the exact cover with $\abs{\mathcal{S}'} = q$.
Without loss of generality, assume that the sets in $\mathcal{S}'$ are $S_1, \dots, S_q$.
We construct the optimal solution $\bm{L}^*$ as: $L_s^* = w_s$ for any $s\in [q]$ and $L_{q+1}^* = 1$.
It can be verified that the linear constraint is satisfied.
The objective value is given by
\begin{align*}
\sum_{s=1}^{q} \Phi_s(L_s^*) + \Phi_{q+1}(L_{q+1}^*) &=\sum_{s=1}^{q} \left(\frac12 + w_s\right) + \sum_{i=1}^n \min(\gamma_{i, q+1}, 1) \\
& = \frac{q}{2} + W_{X3C} + \sum_{i=1}^n \gamma_{i, q+1}.
\end{align*}

Next, if the \xc instance is a no instance, then we show the optimal value of program~(\ref{eq:best-response-reformulated}) is strictly less than $q/2 + W_{X3c} + \sum_{i=1}^n \gamma_{i, q+1}$.
Otherwise, assume that there exists an optimal solution $\bm{L}^*$ with objective value of at least $q/2 + W_{X3c} + \sum_{i=1}^n \gamma_{i, q+1}$.
By \Cref{lem:almost-integeral}, we can construct another optimal solution $\bm{L}^\dag$ such that at least $q-1$ of $L_1^\dag, \dots, L_q^\dag$ are in the set $\{w_1, \dots, w_m\}$; at most one of $L_1^\dag, \dots, L_q^\dag$ is not in the set $\{w_1, \dots, w_m\}$ but still located in some interval $(w_r^-, w_r)$; and $L_{q+1}^\dag \ge 1$.
If all the $L_1^\dag, \dots, L_q^\dag$ are in the set $\{w_1, \dots, w_m\}$, then we can find $q$ sets in $\mathcal{S}$ corresponding to these weights.
It can be verified that these $q$ sets form an exact cover for $U$, which contradicts the assumption.
Otherwise, suppose $L_s^\dag = w_{r(s)}$ for any $s\in [q]\setminus \{j\}$ and $L_j^\dag \in (w_{r}^-, w_{r})$ for some $j\in [q]$ and $r\in [m]$.
Let $\alpha = L_j^\dag / w_{r(j)}$.
Then, we have
\begin{align*}
\Phi_q(L_s^\dag) & =\frac12 + w_{{r(j)}-1} + \abs{\G_r^{r(j)}}\cdot (L_{r(j)}^\dag - w_{r(j)}^-) \\
& =  \abs{\G_r^{r(j)}}\cdot L_{r(j)}^\dag \\
& = \alpha\cdot \left(\frac12 + w_{r(j)}\right)\\
\Phi_{q+1}(L_{q+1}^\dag)& =  \sum_{i=1}^n \min(\gamma_{i, q+1}, L_{q+1}^\dag) \le \sum_{i=1}^n \gamma_{i, q+1}
\end{align*}
Therefore,
\begin{align*}
\sum_{s=1}^{q} \Phi_s(L_s^\dag) + \Phi_{q+1}(L_{q+1}^\dag)
&\le \sum_{s\in [q]\setminus \{j\}} \left(\frac12 + w_{r(s)}\right) + \alpha\cdot \left(\frac12 + w_{{r(j)}}\right) + \sum_{i=1}^n \gamma_{i, q+1} \\
& = \frac12\cdot(q-1 + \alpha) + \sum_{s\in [q]\setminus \{j\}} w_{r(s)} + \alpha\cdot w_{{r(j)}} + \sum_{i=1}^n \gamma_{i, q+1} \\
& \le \frac12\cdot(q-1 + \alpha) + W_{X3C} + \sum_{i=1}^n \gamma_{i, q+1} \tag{as $\sum_s \beta_s L_s^\dag \le 1$} \\
& < \frac{q}{2} + W_{X3C} + \sum_{i=1}^n \gamma_{i, q+1}, \tag{as $\alpha < 1$}
\end{align*}
which contradicts the optimality.
Therefore, the optimal value of the best response is smaller than the target value, and it concludes the reduction.
\end{proof}

\section{Non-Existence of NE under Discretized PLC pricing}
\label{sec:disc-plc-ne-cex}

In this section, we show a data market instance for which no Nash Equilibrium (NE) exists even with discretized PLC pricing.

\begin{example}
\label[example]{ex:no-plc-ne}
Consider 2 sellers and $n \ge 4$ buyers. Each buyer has a budget of 1.
For buyer $n$, we have $\tau_{n,1} = 1$ and $\tau_{n,2} = 0$, so we call her the \emph{picky} buyer.
For each buyer $i \in [n-1]$, we have $\tau_{i,1} = \tau_{i,2} = 1$, so we call them \emph{neutral} buyers.
For a neutral buyer, if the bang-per-buck of the two datasets is the same,
we assume that she will purchase dataset 2 first.

Let $\eps \in (0, 1/3]$ such that $1/\eps$ is an odd integer.
Each dataset is partitioned into $1/\eps$ shards (numbered $1, 2, \ldots, 1/\eps$),
where the $k\Th$ shard has price-per-unit $\eps k$.
Each seller $j \in [2]$ must pick a vector $\ellvec_j \defeq (\ell_{j,1}, \ldots, \ell_{j,1/\eps})$
of shard lengths that sum to 1.
\end{example}

We will show that no NE exists for \cref{ex:no-plc-ne}.
We will prove this by contradiction: we assume that $\ellvec \defeq (\ellvec_1, \ellvec_2)$ is an NE.
We use this assumption to narrow down the structure of $\ellvec$, and then finally derive a contradiction.

Let $x_{j,k}$ be the length of shard $k \in [1/\eps]$ of dataset $j \in [2]$ that is purchased by a neutral buyer.
(So a neutral buyer spends $p_{j,k}x_{j,k}$ on that shard.)
If $x_{j,k} = \ell_{j,k}$, we say that the shard is \emph{fully sold}.
If $x_{j,k} = 0$, the shard is \emph{fully unsold}.
Otherwise, the shard is partially sold.
Note that the shard is both fully sold and fully unsold if $\ell_{j,k} = 0$.

\begin{observation}[first $1/\eps-1$ shards of dataset 1 are fully sold]
\label{thm:s1-every-shard}
For all $k \in [1/\eps-1]$, shard $k$ of dataset 1 is fully sold.
Otherwise seller 1 can reduce $\ell_{1,k}$ to $x_{1,k}$ and increase $\ell_{1,1/\eps}$.
This would increase her revenue, since the picky buyer would pay more.
Since we assumed that $\ell$ is NE, this is not possible, so every shard in $[1/\eps-1]$ is fully sold.
\end{observation}

\subsection{Revenue From Neutral Buyers}

Let $s_j$ be the revenue earned by seller $j$ from each neutral buyer, i.e.,
\[ s_j \defeq \sum_{k=1}^{1/\eps} p_{j,k}x_{j,k}. \]
Since each neutral buyer has a budget of 1, we get $s_1 + s_2 \le 1$.
We now show that neutral buyers exhaust their budget, i.e., $s_1 + s_2 = 1$.
Moreover, shard $1/\eps$ of dataset 1 is fully unsold.

\begin{lemma}
\label{thm:exh-budget}
$s_1 + s_2 = 1$ and $x_{1,1/\eps} = 0$.
\end{lemma}
\begin{proof}
Let $s'_1$ be the revenue from a neutral buyer for the first $1/\eps-1$ shards of dataset 1,
i.e., $s'_1 \defeq s_1 - p_{1,1/\eps}x_{1,1/\eps}$. We will show that $s'_1 + s_2 = 1$,
which simultaneously proves that $s_1 + s_2 = 0$ and $x_{1,1/\eps} = 0$.

Let $S$ be the set of all shards of both datasets, excluding shard $1/\eps$ of dataset 1.
Note that shard $1/\eps$ of dataset 1 is last in the MBB ordering.
Suppose $s'_1 + s_2 < 1$. Then neutral buyers are not exhausting their budgets on $S$,
so shards in $S$ are fully sold.

Let $\delta \defeq (1-s'_1-s_2)/(1-s_2) > 0$.
Suppose seller 2 switches to a strategy $\ellvechat_2$ where she shrinks the length of each shard
by a factor $(1-\delta)$, and then adds the remaining length to shard $1/\eps$. Formally,
$\ellvechat_{2,k} \defeq (1-\delta)\ellvec_{2,k} + \delta\boolOne(k = 1/\eps)$ for all $k \in [1/\eps]$.
The new total price of shards in $S$ then becomes $s'_1 + (1-\delta)s_2 + \delta = s'_1 + s_2 + \delta(1-s_2) = 1$.
Hence, all shards in $S$ continue to be fully sold, and the neutral buyers exhaust their budget.
The sellers' new revenues from each neutral buyer are $s'_1$ and $1-s'_1 > s_2$, respectively.
Thus, seller 2's revenue increases, which contradicts the fact that $\ell$ is an NE.
So, we have $s'_1 + s_2 = 1$, and hence, $s_1 + s_2 = 1$ and $x_{1,1/\eps} = 0$.
\end{proof}

Next, we give lower bounds on $s_1$ and $s_2$.

\begin{lemma}
\label{thm:s-lbs}
$s_1 > 0$ and $s_2 \ge (1+\eps)/2$.
\end{lemma}
\begin{proof}
Suppose $s_1 = 0$. By \cref{thm:exh-budget}, $s_1 + s_2 = 1$, so $s_2 = 1$.
Thus, shard $1/\eps$ of dataset 2 has length 1, i.e., $\ell_{2,1/\eps} = 1$.
If $\ell_{1,k} > 0$ for any $k < 1/\eps$, then $s_1 > 0$, so now assume $\ell_{1,1/\eps} = 1$.
Thus, both sellers have a linear price of 1 for their dataset. Seller 1's total revenue is 1.
If she changes the price to $1-\eps$ (i.e., shard $1/\eps-1$ gets length 1),
then her revenue increases to $n(1-\eps) > 1$. This violates NE, so we cannot have $s_1 = 0$.

Let
\[ t \defeq \frac{1+1/\eps}{2} = \frac{1+\eps}{2\eps}. \]
Suppose seller 2 sets the length of shard $t$ to 1, i.e.,
she switches to $\ellvechat_2$, where $\ellhat_{2,k} \defeq \boolOne(k = t)$.
Seller 1 stays at $\ellvec_1$. Only the first $t-1$ shards of dataset 1 can be
before shard $t$ of dataset 2 in the MBB ordering of neutral buyers.
Thus, the revenue from each neutral buyer is at least $\min(t\eps, 1-(t-1)\eps) = (1+\eps)/2$
(since dataset 2's price is $t\eps$, and the residual budget is $1-(t-1)\eps$).
Since $\ellvec$ is an NE, seller 1's revenue cannot increase with this deviation, so $s_2 \ge (1+\eps)/2$.
\end{proof}

\subsection{Narrowing Sellers' Strategies}

For each seller $j \in [2]$, we have $s_j > 0$ (by \cref{thm:s-lbs}), so $x_{j,k} > 0$ for some $k \in [1/\eps]$.
Define seller $j$'s \emph{last shard to be purchased} (by a neutral buyer)
as $k_j \defeq \max\{k \in [1/\eps]: x_{j,k} > 0\}$.
By \cref{thm:exh-budget}, we have $x_{1,1/\eps} = 0$, so $k_1 \in [1/\eps-1]$.

We now show that $k_2 = k_1 + 1$, and seller 1 has a single shard of positive length in $[1/\eps-1]$.

\begin{lemma}
\label{thm:last-shard}
$k_2 = k_1 + 1$, and $x_{1,k} = 0$ for all $k < k_1$.
\end{lemma}
\begin{proof}
Suppose $\ell_{2,k_1+1} = 0$. Then by reducing $\ell_{1,k_1}$ to 0 and increasing $\ell_{1,k_1+1}$,
seller 1 can increase her revenue (revenue from neutral buyers doesn't decrease, and revenue from picky buyer increases).
This violates NE, so we get $\ell_{2,k_1+1} > 0$.

Suppose $x_{2,k_1+1} = 0$. Then seller 2 can set $\ell_{2,k}$ to 0 for all $k > k_1$ and increase $\ell_{2,k_1}$
(she would be undercutting seller 1 on shard $k_1$).
This increases her revenue which would violate NE, so we get $x_{2,k_1+1} > 0$.
Thus, $k_2 \ge k_1 + 1$.

We now show that if $x_{1,k} > 0$ for some $k < k_2-1$, then seller 1 can increase her revenue
by increasing the length of shard $k_2-1$, which violates NE.
This would prove that $x_{1,k} = 0$ for all $k \in [k_2-2]$, so $k_1 = k_2 - 1$.

Suppose seller 1 switches to strategy $\ellvechat_1$, where
\[ \ellhat_{1,k} \defeq \begin{cases}
1 - \ell_{1,1/\eps} & \text{ if } k = k_2 - 1
\\ \ell_{1,1/\eps} & \text{ if } k = 1/\eps
\\ 0 & \text{ otherwise}
\end{cases}. \]
Let $s'_2$ be the total revenue from a neutral buyer for the first $k_2-1$ shards of dataset 2, i.e.,
\[ s'_2 \defeq \sum_{k=1}^{k_2-1} p_{2,k}x_{2,k}. \]
In $(\ellvec_1, \ellvec_2)$, we had $x_{2,k_2} > 0$, so $s_1 + s'_2 < 1$.
In $(\ellvechat_1, \ellvec_2)$, if shard $k_2$ of dataset 2 is still partially or fully sold,
then shard $k_2-1$ of dataset 1 is fully sold, so seller 1's revenue from a neutral buyer increases.
If shard $k_2$ of dataset 2 is no longer partially or fully sold,
then the buyer is exhausting her budget on shard $k_2-1$ of dataset 1,
so seller 1's revenue from a neutral buyer is $1-s'_2 > s_1$.
In both cases, seller 1's revenue from all buyers increases, which violates NE.
Thus, we get that $x_{1,k} = 0$ for all $k \in [k_2-2]$, so $k_1 = k_2 - 1$.
\end{proof}

We have considerably narrowed down seller 1's strategy.
By combining \cref{thm:last-shard,thm:s1-every-shard}, we get that
shards $k_1$ and $1/\eps$ can be the only ones with positive length,
of which the former is fully sold and the latter is fully unsold.

We now show that some shard of dataset 2 is not fully sold.

\begin{lemma}
\label{thm:s2-not-fully-sold}
Some shard of dataset 2 is not fully sold, i.e.,
$\exists k \in [1/\eps]$ such that $x_{2,k} < \ell_{2,k}$.
\end{lemma}
\begin{proof}
Suppose all shards of dataset 2 are fully sold.
Suppose seller 1 changes the length of shard $1/\eps$ to 1.
Then her revenue from neutral buyers stays the same (by \cref{thm:exh-budget}),
but her revenue from the picky buyer increases, which violates NE.
Thus, some shard of dataset 2 is not fully sold.
\end{proof}

Let us now finally prove a contradiction and show that $\ellvec$ is not an NE.

\begin{lemma}
\label{thm:last-shard-2}
$\ellvec$ is not an NE.
\end{lemma}
\begin{proof}
We start by proving that
$k_1 > (1-\eps)/(2\eps)$ and $\ell_{1,1/\eps} = 0$.

We have $s_1 = k_1\eps\ell_{1,k_1} \le k_1\eps$.
By \cref{thm:s2-not-fully-sold}, we get $s_2 < k_2\eps$.
By \cref{thm:exh-budget}, $s_1 + s_2 = 1$, so $1 < (k_1 + k_2)\eps$.
By \cref{thm:last-shard}, $k_2 = k_1+1$, so $k_1 > (1-\eps)/(2\eps)$.

Suppose $\ell_{1,1/\eps} > 0$ and seller 1 deviates to strategy $\ellvechat_1$, where $\ellhat_{1,k_1} \defeq 1$.
Then the increase in her revenue is
\[ nk_1\eps - (nk_1\eps\ell_{1,k} + \ell_{1,1/\eps})
    = (nk_1\eps-1)\ell_{1,1/\eps} \ge \left(n\frac{1-\eps}{2} - 1\right)\ell_{1,1/\eps} > 0. \]
This violates NE, so $\ell_{1,1/\eps} = 0$.

By \cref{thm:exh-budget,thm:s-lbs}, we get $s_1 \le (1-\eps)/2$.
Hence, $s_1 = k_1\eps \le (1-\eps)/2$, so $k_1 \le (1-\eps)/(2\eps)$.
This is a contradiction. Thus, $\ellvec$ is not an NE.
\end{proof}

\section{Sparsity of PLC Best Response}
\label{sec:no-of-shards-extra}

\begin{lemma}
Consider a data market instance $([n], [m], (u_i)_{i=1}^n, (b_i)_{i=1}^n)$.
Each seller $j \in [m]$ is given a set $P_j \defeq \{p_{j,1}, \ldots, p_{j,m_j}\}$ of prices,
where $p_{j,1} < \ldots < p_{j,m_j}$, and must decide the size $\ell_{j,k}$ of the shard having price $p_{j,k}$.

For any seller $j$, given the pricing strategies $\ellvec_{-j}$ of the other sellers,
there is a best-response strategy $\ellvec_j$ where $|\supp(\ellvec_j)| \le n$.
\end{lemma}
\begin{proof}
Fix a seller $j$ and the strategy profile $\ellvec_{-j}$ of the remaining sellers.
We will show how to transform any strategy $\ellvec_j$ into another strategy $\ellhat_j$
such that $|\supp(\ellhat_j)| \le n$ and seller $j$'s revenue does not decrease,
i.e., $r_j(\ellvec) \le r_j(\ellhat_j, \ellvec_{-j})$.

We will do the transformation in multiple steps.
In each step, we apply one of two types of transformations.

\paragraph{Type 1:}
For any $k_1, k_2 \in [m_j]$, the $(k_1, k_2)$-type-1 transformation is applicable when
\begin{tightenum}
\item $k_1$ and $k_2$ are adjacent in $\supp(\ellvec_j)$, i.e.,
    $\ell_{j,k_1} > 0$, $\ell_{j,k_2} > 0$, and $\ell_{j,k} = 0$ for all $k \in [k_1 + 1, k_2 - 1]$.
\item Some buyer $i$ purchases at least part of shard $k_2$.
\item Every buyer $i$ either doesn't buy any part of $k_1$,
    or finds shard $k_2$ desirable and completely buys every shard (of positive size, and of any seller)
    in her bang-per-buck order that appears before $k_2$.
\end{tightenum}
In this transformation, we transfer the size of shard $k_1$ to shard $k_2$, i.e.,
if $\ellhat_j$ is the new pricing strategy, then $\ellhat_{j,k_1} = 0$,
$\ellhat_{j,k_2} = \ell_{j,k_1} + \ell_{j,k_2}$,
and $\ellhat_{j,k} = \ell_{j,k}$ for all $k \in [m_j] \setminus \{k_1, k_2\}$.
One can see that a $(k_1, k_2)$-type-1 transformation doesn't decrease revenue from any buyer
by laboriously listing out the possibilities of where along her bang-per-buck order a buyer may stop buying.

\paragraph{Type 2:}
A type-2 transformation is applicable if $|\supp(\ellvec_j)| > 1$ and the last shard (of positive size) has no revenue,
i.e., $\gamma'_{i,j,k_2}(\ellvec) = 0$ for all $i \in [n]$, where $k_2 = \max(\supp(\ellvec_j))$.
In this transformation, we transfer the size of the last shard to the second-last shard.
Formally, let $k_1 \defeq \max(\supp(\ellvec_j) \setminus \{k_2\})$.
Then we set $\ellhat_{j,k_1} = \ell_{j,k_1} + \ell_{j,k_2}$, $\ellhat_{j,k_2} = 0$,
and $\ellhat_{j,k} = \ell_{j,k}$ for all $k \in [m_j] \setminus \{k_1, k_2\}$.
One can see that a type-2 transformation doesn't decrease revenue from any buyer.

We apply as many type-1 transformations as possible, and then apply as many type-2 transformations as possible.
One can show that no transformations are possible after this.
The inapplicability of type-2 transformations implies that in $\ellhat_j$,
some buyer purchases at least a bit of every shard.
The inapplicability of type-1 transformations implies that
for any $k_1$ and $k_2$ adjacent in $\supp(\ellhat_j)$,
some buyer purchases (at least a part of) $k_1$ but doesn't purchase everything just before $k_2$.
Thus there must be at least $|\supp(\ellhat_j)|$ buyers.
\end{proof}

\section{CE Implies NE in Rivalrous Settings}
\label{sec:ce-is-ne-rival}

In this section, we provide a proof sketch showing that CE prices form a NE in classical markets with rivalrous goods and linear buyer utilities.

\begin{proof}[Proof sketch]
At CE prices $\bm p$, all goods are completely sold, so the revenue of each seller $j$ is $p_j \cdot s_j$, where $s_j$ is the supply of good $j$. Clearly, there is no incentive for any seller to lower her prices, as this would only decrease her revenue. Now, consider the case where seller $j$ increases the price of her good. We examine two scenarios depending on whether a buyer was purchasing good $j$ at price $\bm{p}$.

Case 1: The buyer was not purchasing good $j$ at price $\bm{p}$. In this case, the buyer will not be interested in purchasing good $j$ at the new, higher price.

Case 2: The buyer was purchasing good $j$ at price $\bm{p}$.  Due to the price increase, the buyer will strictly prefer other goods they were already purchasing at price $\bm{p}$. As a result, the total money available for good $j$ either remains the same or decreases.

In both cases, the total revenue of seller $j$ at the higher price is non-increasing. This implies that CE prices form a Nash equilibrium.
\end{proof}

\end{document}